\documentclass[10pt, journal,twocolumn]{IEEEtran}

\usepackage[utf8]{inputenc}
\usepackage{amsmath}
\usepackage{subfig}
\usepackage[dvips]{graphicx}
\usepackage{amsmath}
\usepackage{amsmath}
\usepackage{mathtools}
\usepackage{times}
\usepackage{tabulary}
\usepackage{amssymb}
\usepackage{amsthm}
\usepackage{amsmath}
\usepackage{multirow}
\usepackage{hyperref}
\usepackage{color}
\usepackage{cite}
\usepackage[update,prepend]{epstopdf}

\graphicspath{ {./figs/} }
\newcounter{MYtempeqncnt}

\newcommand{\expect}[1]{{\mathbb{E}\left[{#1}\right]}}
\newcommand{\cexpect}[2]{{\mathbb{E}_{#2}\left[{#1}\right]}}
\newcommand{\pr}{{\mathbb{P}}}

\newcommand{\pgf}{\mathsf{G}}
\newcommand{\mgf}{\mathsf{M}}
\newcommand{\stirling}[2]{\mathsf{S}(#1,#2)}
\newcommand{\set}[1]{{\mathcal{#1}}}
\newcommand{\setopen}[2]{{\mathcal{#1}^o_{#2}}}
\newcommand{\setclose}[2]{{\mathcal{#1}^c_{#2}}}

\newcommand{\ple}[2]{{\alpha_{#2}}}
\newcommand{\nple}[2]{{\hat{\alpha}_{#2}}}
\newcommand{\power}[2]{\mathrm{P}_{#1#2}}
\newcommand{\res}[2]{\mathrm{W}_{#1#2}}
\newcommand{\npower}[2]{\mathrm{\hat{P}}_{#1#2}}
\newcommand{\bias}[2]{\mathrm{B}_{#1#2}}
\newcommand{\noisepower}[1]{{\sigma^2_{#1}}}

\newcommand{\nbias}[2]{\mathrm{\hat{B}}_{#1#2}}

\newcommand{\dnstysrat}[2]{{\lambda_{#1#2}}}
\newcommand{\dnstycrat}[2]{{{\lambda}_{#1#2^{'}}}}

\newcommand{\tierdnsty}[2]{{\lambda}_{#1#2}}
\newcommand{\userdnsty}{\lambda_u}
\newcommand{\R}{{\mathbb{R}}}
\newcommand{\SINRthresh}{\tau}
\newcommand{\Z}{\mathsf{Z}}

\newcommand{\RATEthresh}{\rho}
\newcommand{\nRATEthresh}{\hat{\rho}}

\newcommand{\uRATEthresh}{t}
\newcommand{\metric}[2]{\mathrm{T}_{#1#2}}
\newcommand{\nmetric}[2]{\hat{\mathrm{T}}_{#1#2}}

\newcommand{\SINR}{\mathtt{SINR}}
\newcommand{\SIR}{\mathtt{SIR}}
\newcommand{\SNR}{\mathtt{SNR}}
\newcommand{\assocr}{\mathcal{C}}
\newcommand{\rate}[2]{R_{#1#2}}
\newcommand{\ndist}[2]{z_{#1#2}}
\newcommand{\ndistns}{z}
\newcommand{\NDIST}[2]{Z_{#1#2}}

\newcommand{\ndistnsc}{y}
\newcommand{\NDISTc}[2]{Y_{#1#2}}
\newcommand{\NAP}[2]{X_{#1#2}^*}
\newcommand{\PPPsrat}[2]{\Phi_{#1#2}}

\newcommand{\PPPcrat}[2]{{\Phi}_{#1#2^{'}}}
\newcommand{\PPPu}{\Phi_{u}}
\newcommand{\tierPPP}[2]{{\Phi}_{#1#2}}
\newcommand{\load}[2]{{N}_{#1#2}}
\newcommand{\oload}[2]{{N}_{o,#1#2}}

\newcommand{\avload}[2]{\bar{N}_{#1#2}}
\newcommand{\mk}{(m,k)}
\newcommand{\area}{C}
\newcommand{\areans}{c}

\newcommand{\pcov}{\mathcal{S}}
\newcommand{\psys}{\mathcal{R}}
\newcommand{\passoc}{\mathcal{A}}

\newtheorem{thm}{{\bf Theorem}}
\newtheorem{cor}{Corollary}
\newtheorem{lem}{Lemma}
\newtheorem{prop}{Proposition}
\theoremstyle{definition}
\newtheorem{definition}{Definition}

\theoremstyle{remark}
\newtheorem{rem}{Remark}

\begin{document}
\title{Offloading in Heterogeneous Networks: Modeling, Analysis, and Design Insights}
\author{Sarabjot Singh,~\IEEEmembership{Student~Member,~IEEE,} Harpreet S. Dhillon,~\IEEEmembership{Student~Member,~IEEE,}\\ and Jeffrey G. Andrews,~\IEEEmembership{Fellow,~IEEE} \thanks{This work has been supported by the Intel-Cisco Video Aware Wireless Networks (VAWN) program and NSF grant CIF-1016649. A part of this paper is accepted for presentation at IEEE ICC 2013 in Budapest, Hungary \cite{SinDhiAndICC13}.

 The  authors  are  with  Wireless Networking and Communications Group (WNCG),  The  University  of  Texas at  Austin  (email:  sarabjot@utexas.edu,  dhillon@utexas.edu, jandrews@ece.utexas.edu).}}
\maketitle
\begin{abstract}
Pushing data traffic from cellular to WiFi is an example of inter radio access technology (RAT) offloading. While this clearly alleviates congestion on the over-loaded cellular network, the ultimate potential of such offloading and its effect on overall system performance is not well understood. To address this, we develop a general and tractable model that consists of $M$ different RATs, each deploying up to $K$ different tiers of access points (APs), where each tier differs in transmit power, path loss exponent, deployment density and bandwidth.  Each class of APs is modeled as an independent Poisson point process (PPP), with mobile user locations modeled as another independent PPP, all channels further consisting of i.i.d. Rayleigh fading.  The distribution of rate over the entire network is then derived for a weighted association strategy, where such weights can be tuned to optimize a particular objective.
We show that the optimum fraction of traffic offloaded to maximize $\SINR$  coverage is not in general the same as the one that maximizes rate coverage, defined as the fraction of  users achieving a given rate.
\end{abstract}

\section{Introduction}\label{sec:intro}
Wireless networks are facing  explosive data demands driven largely by video.
While operators continue to rely on their (macro) cellular networks to provide wide-area coverage, they are eager  to find complementary alternatives to ease the pressure, especially in areas where subscriber density is high. Complementing  the fast  evolving heterogeneous cellular networks (HCNs) \cite{qcom_hetnet_wmag} with the already widely deployed  WiFi APs is very attractive to operators and a key aspect of their strategy \cite{qcom_wifi_hetnet}. In fact, WiFi access points (APs) along with  femtocells  are projected to carry over 60\% of all the  global data traffic  by 2015 \cite{juniper_report}. A future wireless heterogeneous  network (HetNet) can be envisioned to have operator-deployed macro base stations (BSs) providing a coverage blanket,  along with pico BSs, low powered user-deployed femtocells  and  user/operator-deployed low powered WiFi APs.
\subsection{Motivation and Related Work}
Aggressively offloading mobile users from macro BSs to smaller BSs like  WiFi hotspots, however,  can lead to degradation of user-specific as well as network wide performance.
For example, a  WiFi AP with excellent signal strength may suffer from heavy load  or have less effective bandwidth (channels), thus reducing the effective rate it can serve at \cite{qcom_wifi_offload}. On the other hand, a conservative approach may result in load disparity,  which not only leads to underutilization of resources  but also degrades the performance of multimedia applications due to bursty interference caused by the lightly loaded APs \cite{singh2012interference}.  Clearly, in such cases any offloading strategy agnostic to these conditions is undesirable, which emphasizes the importance of more adaptive offloading strategies.

RAT selection  has been studied extensively in earlier works  both  from centralized as well as decentralized aspects (see \cite{mathematicalKuoWang12} for a survey). Fully centralized schemes, such as in \cite{navlinwong08,prekum06,kumaltkel07itc},  try to maximize a network wide utility as a solution to the association optimization problem. Decentralized schemes have been studied from game theoretic approaches in \cite{kumaltkel07,KhaIbrCohLahToh11,elaalthadalt10} and as heuristic randomized algorithms in \cite{moeibrlahkha12}. However most of these works focused on  flow level assignment and  lacked explicit spatial location modeling of the APs and users and the corresponding impact on association. The presented work  is more similar to ``cell breathing"  \cite{TogYosKoh98,Jal98,SanWanMadGit04}, wherein the BS association regions are expanded or shrunk depending on the load. Contemporary cellular standards like LTE use a cell breathing approach to address the problem of load balancing in HCNs through cell range expansion (CRE) \cite{3ggp_r1_103824, qcom_hetnet_wmag} where  users are offloaded to smaller cells using an association bias. A positive association bias implies that a user would be offloaded to a smaller BS as soon as the received power difference from the macro and  small BS drops  below the bias value. The presented work employs CRE to tune the aggressiveness of offloading from one RAT to another in HetNets.
Tools from Poisson point process (PPP) theory  and stochastic geometry \cite{mecke_book} allow us to quantify the optimal association bias  of each constituent tier of each RAT, which maximizes the fraction of time a typical user in the network is served with a rate greater than its minimum rate requirement.

The metric of rate coverage used in this paper, which  signifies the fraction of user population able to meet their rate thresholds, captures the inelasticity of traffic such as video services \cite{singhvideoicc12}, whereas traditional utility based metrics are more suitable for elastic traffic with no hard rate thresholds.
There has been considerable advancement  in the theory of  HCNs \cite{dhiganbacand12,josanxiaand12,mukh12} whereby the locations of APs of each tier are assumed to form a homogeneous PPP.  The case of modeling macro cellular networks using a  PPP has been strengthened through empirical validation in \cite{andganbac11} and theoretical validation in \cite{BlaKarKee12}. Load distribution was derived for macro cellular networks in \cite{YuKim11} and an empirical fitting based approach was proposed  in \cite{CaoZhoNiu12} for association area distribution in a two-tier cellular network.  See \cite{singh2012interference,DhiGanJ2013,lin2012modeling} for a spectral efficiency analysis, where load is modeled through  activity of AP queues. While the PPP assumption offers attractive tractability in modeling interference and hence the signal-to-interference-and-noise ratio ($\SINR$) in HetNets, the distribution of \textit{rate} has been elusive.
Superposition of point processes, each denoting a class\footnote{A class refers to a distinct RAT-tier pair.} of APs, leads to the formation of disparate association regions (and hence load distribution) due to the unequal transmit powers, path loss exponents  and association weights among different classes of APs.
Thus, resolving to complicated system level simulations for investigating impact of various wireless algorithms on rate, even for preliminary insights, is not uncommon.
One of the goals of this paper is to bridge this gap and provide a tractable framework for deriving the  rate distribution in  HetNets.
\subsection{Contributions}
The contributions of this paper can be categorized under two main headings.
\begin{enumerate}
\item{\textbf{Modeling and Analysis.}}
A  general $M$-RAT $K$-tier HetNet model  is proposed  with each class of  APs  drawn from a  homogeneous PPP.  This is similar to \cite{dhiganbacand12,josanxiaand12,mukh12} with the key difference  being  the  APs of a RAT act as interferers to only the  user associated with that RAT. For example, cellular BSs do not interfere with the  users associated with a WiFi AP and vice versa.   The proposed model is validated by comparing the analytical  results  with those of a realistic multi-RAT deployment in Section \ref{sec:validation}. \\
\textit{\textbf{Association Regions in HetNet:}}
Based on the weighted path loss based association used in this work, the tessellation formed by association regions of APs (region served by the AP)  is characterized as a general form of the multiplicatively weighted Poisson Voronoi (PV).  Much progress has been made in modeling the  area of Poisson Voronoi, see \cite{Gil62, HinMil80,Ferenc2007} and references therein, however that of a general multiplicatively weighted PV  is an  open problem. We propose an analytic approximation for characterizing the association areas (and hence the load) of an AP,  which is shown to be quite accurate in the context of rate coverage.\\
\textit{\textbf{Rate Distribution in HetNet:}}
We derive the rate complementary cumulative distribution function  (CCDF) of a typical user in the presented HetNet setting in Section \ref{sec:cov}.  Rate distribution incorporates congestion in addition to the  proximity effects that may not be accurately captured by the $\SINR$ distribution alone.  Under certain plausible scenarios the derived expression is in closed form and provides insight into system design.
\item{\textbf{System Design Insights.}}
This work allows the inter-RAT offloading to be seen through the prism of  association bias wherein the bias can be tuned to suit a network wide objective. We present the following insights in Section \ref{sec:opt_offload} and \ref{sec:results}. \\
\textit{\textbf{$\SINR$ Coverage:}}
The probability that a randomly located user has $\SINR$ greater than an arbitrary threshold is called $\SINR$ coverage; equivalently this is the CCDF of $\SINR$. In a simplified  two-RAT scenario, e.g. cellular and WiFi, it is shown that the optimal amount of traffic to be offloaded, from one to another, depends solely on their respective  $\SINR$ thresholds. The optimal association bias, however, is shown to be inversely proportional to the density and transmit power of the corresponding RAT. The maximum $\SINR$ coverage  under the optimal association bias is then shown to be independent of the density of  APs in the network. \\
\textit{\textbf{Rate Coverage:}}
The probability that a randomly located user has rate greater than an arbitrary threshold is called rate coverage; equivalently this is the CCDF of rate. We show that
 the  amount of  traffic to be routed through a RAT for maximizing rate coverage can be found analytically and  depends on the ratio of the respective resources/bandwidth  at each RAT and the user's respective rate (QoS) requirements. Specifically, higher the corresponding ratio, the more traffic should be  routed through the corresponding RAT. Also, unlike $\SINR$ coverage,  the optimal traffic offload fraction increases with the density of the corresponding RAT. Further, the rate coverage always increases with the density of the infrastructure.
\end{enumerate}

\section{System Model}\label{sec:sysmodel}
The system model in this paper considers up to a $K$-tier deployment of the APs for each of the $M$-RATs. The set of APs belonging to the same RAT operate in the same spectrum and hence do not interfere with the APs of other RATs.
The locations of the APs of the $k^{\mathrm{th}}$ tier of the $m^{\mathrm{th}}$ RAT are modeled as a 2-D homogeneous PPP, $\PPPsrat{m}{k}$, of density (intensity) $\dnstysrat{m}{k}$.
 Also, for every class $(m,k)$ there might be BSs allowing no access (closed access) and thus acting only as interferers. For example, subscribers of a particular operator are not able to connect to another operator's  WiFi APs but receive interference from them. Such  closed access APs are modeled
as an independent tier ($k'$) with PPP $\PPPcrat{m}{k}$  of density $\dnstycrat{m}{k}$. The set of all such pairs with non-zero densities in the network is denoted by $\set{V} \triangleq\bigcup_{m=1}^M\bigcup_{k\in \set{V}_m} (m,k) $ with $\set{V}_m$ denoting the set of all the tiers of RAT-$m$, i.e., $\set{V}_m = \{k:\tierdnsty{m}{k}+\dnstycrat{m}{k}\neq 0\}$. Similarly,  $\setopen{V}{m}$ and $\setclose{V}{m}$ is used to denote the set of open and closed access tiers of RAT-$m$, respectively. Further, the set of open access classes of APs is $\setopen{V}{}\triangleq\bigcup_{m=1}^M \bigcup_{k\in \setopen{V}{m}}(m,k)$.
The users in the network are assumed to be distributed according to an independent homogeneous PPP $\PPPu$ with density $\userdnsty$.

Every AP of $(m,k)$ transmits with the same transmit power $\power{m}{k}$ over  bandwidth $\res{m}{k}$.
The downlink desired and interference signals are assumed to experience path loss with a path loss exponent $\ple{m}{k}$  for the corresponding tier $k$. The power received at a user  from an AP  of $\mk$ at a distance $x$ is  $\power{m}{k} h x^{-\ple{m}{k}}$ where $h$ is the channel power gain.  The random channel gains are  Rayleigh distributed with average power of unity, i.e., $h \sim \exp(1)$. The general fading distributions can be considered at some loss of tractability \cite{BacBlaMuh09}. The noise is assumed additive with power $\noisepower{m}$ corresponding to the $m^{\text{th}}$ RAT.   Readers can refer to Table \ref{table:notationtable} for quick access to the notation used in this paper. In the table and the rest of the paper, the  normalized value of a parameter of a class is its value divided by the value it takes for the class of the serving AP.
\begin{table}
	\centering
\caption{Notation Summary}
	\label{table:notationtable}
  \begin{tabulary}{\columnwidth}{ |c | C | }
    \hline
    \textbf{Notation} & \textbf{Description} \\ \hline
    $M$ & Maximum number of RATs  in the network\\ \hline
 $K$ & Maximum number of tiers of a RAT \\ \hline
$(m,k)$  &  Pair denoting the $k^\mathrm{th}$ tier of the $m^\mathrm{th}$ RAT\\ \hline
$\set{V}; \setopen{V}{}$ & The set of classes of APs $\bigcup_{m=1}^M\bigcup_{k \in\set{V}_m} (m,k)$, where $\set{V}_m = \{k: \tierdnsty{m}{k}+\dnstycrat{m}{k}\neq 0\}$; the set of open access classes of APs $\bigcup_{m=1}^M\bigcup_{k \in\setopen{V}{m}}(m,k)$,  where $\setopen{V}{m} = \{k: \tierdnsty{m}{k}\neq 0\}$ \\ \hline
$\PPPsrat{m}{k}; \PPPcrat{m}{k}; \PPPu$  & PPP of the open access APs of $(m,k)$;  PPP of the closed access APs of  $\mk$; PPP of the mobile users  \\ \hline
$\dnstysrat{m}{k}; \dnstycrat{m}{k}; \userdnsty$ & Density of open access APs of $(m,k)$; density of closed access APs of $(m,k)$; density of mobile users \\ \hline
%$\NDIST{i}{j}$& Distance of the nearest BS in $(i,j)$\\\hline
$\metric{m}{k}; \nmetric{m}{k}$ & Association weight for $(m,k)$; normalized (divided by that of the serving AP) association weight for $\mk$\\\hline
$\power{m}{k}; \npower{m}{k}$ & Transmit power of APs of $\mk$, specifically $\power{m}{1} = 53$ dBm, $\power{m}{2} = 33$ dBm, $\power{m}{3} = 23$ dBm; normalized transmit power of APs of $\mk$ \\\hline
%$ \plist $ & Preference list for the choice of network in each MRAT BS for each user\\\hline
$\bias{m}{k}; \nbias{m}{k}$ & Association bias for $(m,k)$; normalized association bias  for $\mk$.\\\hline
$\ple{m}{k}; \nple{m}{k}$ & Path loss exponent of $k^{\text{th}}$ tier; normalized path loss exponent of $k^{\text{th}}$ tier\\\hline
$\noisepower{m}$ & Thermal noise power corresponding to $m^{\text{th}}$ RAT\\\hline
$\res{m}{k}$ & Effective bandwidth  at an AP of $(m,k)$ \\\hline
$\SINRthresh_{mk}$ & $\SINR$  threshold of user when associated with $\mk$ \\\hline
$\RATEthresh_{mk}$ & Rate threshold of user when associated with $\mk$ \\\hline
$\load{m}{k}$ & Load (number of users) associated with an AP of $\mk$\\\hline
$\area_{mk}$ & Association area of a typical AP of $\mk$\\\hline
$\pcov_{mk}; \pcov$ &  $\SINR$ coverage of user when associated with $\mk$; overall $\SINR$  coverage of  user\\\hline
$\psys_{mk}; \psys$ &  Rate  coverage of user when associated with $\mk$; overall rate coverage of user\\\hline
\end{tabulary}
\end{table}

\subsection {User Association}
For the analysis that follows, let $\NDIST{m}{k}$ denote the distance of a typical user from the nearest AP of $(m,k)$. In this paper, a general association metric is used in  which a mobile user is connected to a particular RAT-tier pair $(i,j)$ if
\begin{equation}\label{eq:association}
(i,j) = \arg \max_{(m,k)\in \setopen{V}{}} \metric{m}{k}\NDIST{m}{k}^{-\ple{m}{k}},
\end{equation}
where $\metric{m}{k}$ is the association weight for $(m,k)$ and ties are broken arbitrarily. These association weights  can  be tuned to suit a certain network-wide objective. As an example,  if  $\metric{1}{k} \gg \metric{2}{k}$, then more traffic is routed through RAT-1 as compared to RAT-2.
Special cases  for the   choice of association weights, $\metric{m}{k}$, include:
\begin{itemize}
\item $\metric{m}{k} = 1$: the association is to the nearest base station.
\item $\metric{m}{k} =  \power{m}{k}\bias{m}{k}$: is the cell range expansion (CRE) technique [2] wherein  the association is based on the maximum biased received power, with $\bias{m}{k} $  denoting the association bias corresponding to  $(m,k)$.
\item Further, if  $\bias{m}{k} \equiv 1$,  then the association is based on  maximum received power.
\end{itemize}
Note that ``$\equiv$" is  henceforth used to assign the same value to a parameter for all classes of APs, i.e.,  $x_{mk}\equiv c$ is equivalent to $x_{mk} =c \,\,\forall\,\, \mk \in \set{V}$.
The optimal association weights maximizing rate coverage  would depend on load, $\SINR$, transmit powers, densities, respective bandwidths, and  path loss exponents of AP classes in the network. Further discussion on the design of optimal association weights  is deferred to Section \ref{sec:opt_offload}.
 For notational brevity the  normalized parameters of $(m,k)$, conditioned on $(i,j)$ being the serving class, are 
\[ \nmetric{m}{k} \triangleq \frac{\metric{m}{k}}{\metric{i}{j}},\,\, \npower{m}{k} \triangleq \frac{\power{m}{k}}{\power{i}{j}},\,\, \nbias{m}{k} \triangleq \frac{\bias{m}{k}}{\bias{i}{j}}\,\,,\,\, \nple{m}{k} \triangleq \frac{\ple{m}{k}}{\ple{i}{j}}.  \]

The association model described above  leads to the formation of association regions  in the Euclidean plane as described below.
\begin{definition}
\textbf{Association region} of an AP is the region of the Euclidean plane in which all users are served by the corresponding AP.
Mathematically, the association region of an AP of class $(i,j)$ located at $x$ is
\begin{multline}
\assocr_{x_{ij}}= \Bigg\{ y \in \R^2: \|y-x\| \leq \left(\frac{\metric{i}{j}}{\metric{m}{k}}\right)^{1/\ple{i}{j}}\|y-\NAP{m}{k}(y)\|^{\nple{m}{k}}\\
 \forall\,\, (m,k) \in \setopen{V}{}\Bigg\},
\end{multline}
where $\NAP{m}{k}(y) = \arg \min\limits_{x \in \PPPsrat{m}{k}}\|y-x\|$.
\end{definition}
The readers familiar with the field of spatial tessellations would recognize that the random tessellation formed by the collection $\{\assocr_{x_{ij}}\}$ of association regions is  a general case of the circular Dirichlet tessellation \cite{AshBol86}.  The circular Dirichlet tessellation (also known as multiplicatively weighted Voronoi) is the special case of the presented model  with equal path loss coefficients.
Fig. \ref{fig:arexpansion} shows the association regions with two classes of APs in the network ($\set{V}=\{(1,1);(2,3)\}$, say) for two ratios of association weights  $\frac{\metric{1}{1}}{\metric{2}{3}}=20$ dB  and  $\frac{\metric{1}{1}}{\metric{2}{3}}=10$ dB. The path loss exponent is $\ple{m}{k}\equiv 3.5$.
\begin{figure}
	\centering
		\includegraphics[width=\columnwidth]{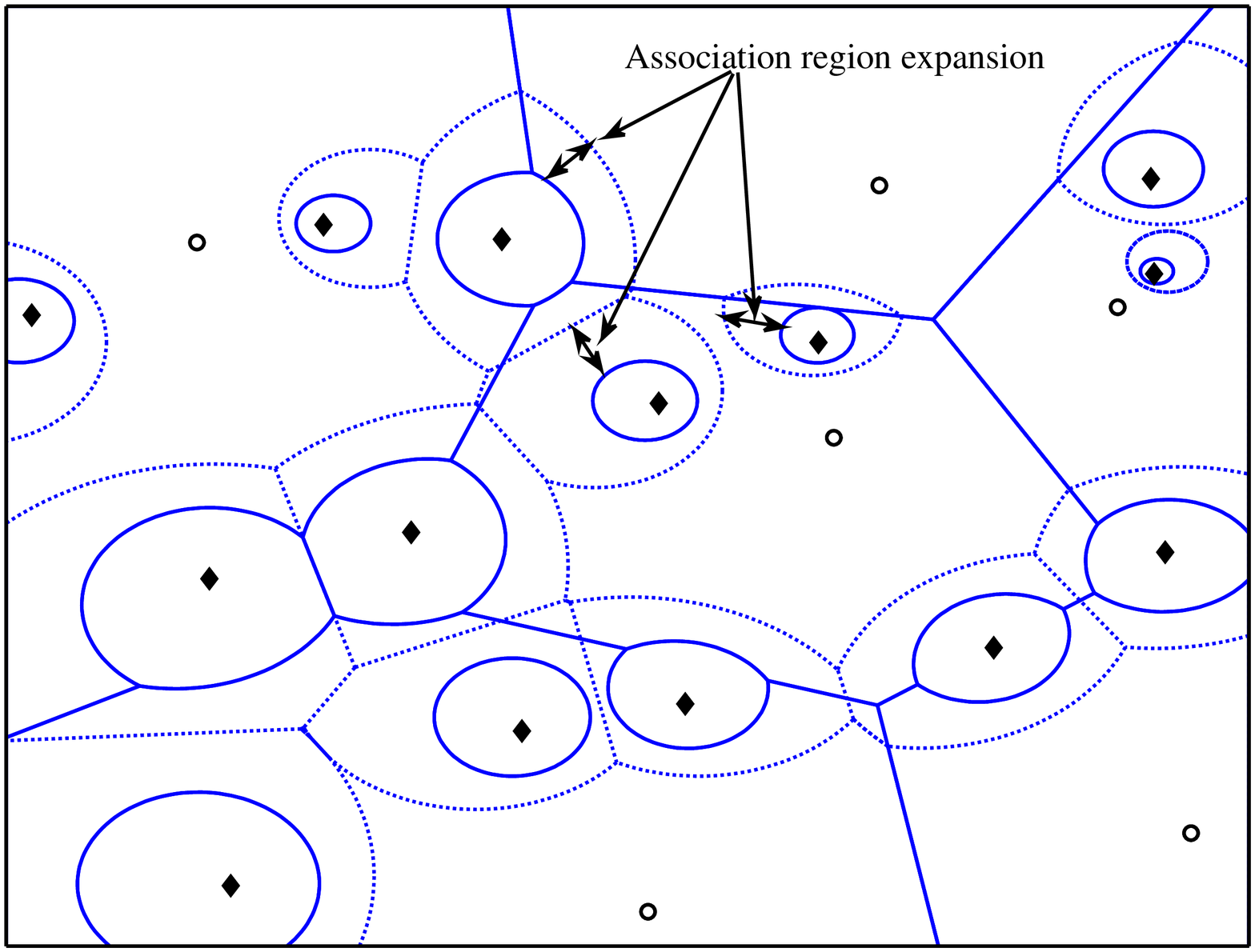}
		\caption{Association regions of a network with $\set{V}=\{(1,1); (2,3)\}$. The APs of $(1,1)$ are shown as hollow circles and those of $(2,3)$ are shown as solid diamonds. Solid lines show the association regions with $\frac{\metric{1}{1}}{\metric{2}{3}}= 20$ dB and dotted lines show the expanded association regions of $(2,3)$ resulting from the use of $\frac{\metric{1}{1}}{\metric{2}{3}}= 10$ dB.}
	\label{fig:arexpansion}
\end{figure}
\subsection{Resource Allocation}
%\begin{asmptn}\label{asmptn:rrm}
 A \textit{saturated} resource allocation model is assumed in the downlink of all the APs. This assumption implies that each AP always has data to transmit to its associated mobile users and hence users can be allocated more rate than their rate thresholds.
Under the assumed resource allocation, each user receives rate proportional to its link's spectral efficiency.  Thus, the  rate of a user associated with $(i,j)$ is given by
\begin{equation}\label{eq:ratemodel}
\rate{i}{j} = \frac{\res{i}{j}}{\load{i}{j}}\log\left(1+\SINR_{ij}\right),
\end{equation}
where  $\load{i}{j}$ denotes the total number of users served by the AP,   henceforth referred to as the \textit{load}.
The presented rate model captures both the congestion effect (through load) and proximity effect (through $\SINR$). For 4G cellular systems,  this rate allocation model has the interpretation of scheduler allocating the OFDMA resources ``fairly" among users. For 802.11 CSMA networks, assuming equal channel access probabilities \cite{kumaltkel07itc,moeibrlahkha12} across associated users, leads to the rate model (\ref{eq:ratemodel}).  Although the above mentioned resource allocation strategy is assumed in the paper, the ensuing analysis can be extended to a RAT-specific resource allocation methodology as well.
\section{Rate Coverage}\label{sec:cov}
This section derives the rate coverage and is the main technical section of the paper. The rate coverage is defined as
\begin{equation}
\psys\triangleq\pr(R>\RATEthresh),
\end{equation}
and can be thought of equivalently as: (i) the probability that a randomly chosen user can achieve a
target rate $\RATEthresh$, (ii) the average fraction of users in the network that  achieve rate $\RATEthresh$, or (iii) the average fraction of the network area that is receiving rate greater than  $\RATEthresh$.

\subsection{Load Characterization}
This section analyzes the load, which is crucial to get a handle on the rate distribution.
The following analysis uses the notion of typicality, which is made rigorous  using Palm theory \cite[Chapter~4]{mecke_book}.

\begin{lem}\label{lem:userpgf}
The load at a typical AP of $(i,j)$  has the  probability generating function (PGF) given by
\begin{equation}
\pgf_{\load{i}{j}}(z)= \expect{\exp\left(\userdnsty\area_{ij}\left(z-1\right)\right)},
\end{equation}
where $\area_{ij}$ is the association area of a typical AP of $(i,j)$.
\end{lem}
\begin{IEEEproof}
We consider the process $\PPPsrat{i}{j}\cup\{0\}$ obtained by adding an AP of $(i,j)$ at the origin of the coordinate system, which is the typical AP under consideration. This is allowed by Slivnyak's theorem \cite{mecke_book}, which states that the properties observed by a typical\footnote{The term typical and random  are interchangeably used in this paper.} point of the PPP, $\Phi_{ij}$, is same as those observed by the point at origin  in the process $\Phi_{ij}\cup\{0\}$.
The random variable (RV) $\load{i}{j}$   is the number of users from $\PPPu$ lying  in the association region  $\assocr_{0_{ij}}$  of the typical cell constructed from the process $\PPPsrat{i}{j}\cup\{0\}$. Letting $\area_{ij}$ denote the random area of this typical association region, the PGF of $\load{i}{j}$ is  given by
\[\pgf_{\load{i}{j}}(z)= \expect{z^{\load{i}{j}}}=\expect{\exp\left(\userdnsty\area_{ij}\left(z-1\right)\right)}\,,\]
where the property used is that conditioned on $\area_{ij}$,  $\load{i}{j}$  is a Poisson RV with mean $\userdnsty\area_{ij}$.
\end{IEEEproof}
As per the association rule (\ref{eq:association}), the probability that a typical user associates with a particular RAT-tier pair would be directly proportional to the corresponding AP density and association weights. The following lemma identifies the exact relationship.
\begin{lem}\label{lem:aspr}
The probability that a typical user is associated with $(i,j)$  is given by
\begin{equation}\label{eq:aspr}
\passoc_{ij} =  2\pi\dnstysrat{i}{j}\int_{0}^{\infty}{z \exp\left(-\pi \sum_{(m,k) \in \setopen{V}{}} G_{ij}(m,k)z^{2/\nple{m}{k}}\right)}\mathrm{d}z,
\end{equation}
where
\begin{equation}G_{ij}(m,k)=\dnstysrat{m}{k} \nmetric{m}{k}^{2/\ple{m}{k}}. \end{equation}
If $\ple{m}{k}\equiv\alpha$, then the association probability is simplified to
\begin{equation}\label{eq:simplifiedassocpr}
\passoc_{ij} =  \frac{\dnstysrat{i}{j}}{\sum_{(m,k) \in \setopen{V}{}} G_{ij}(m,k)}.
\end{equation}
\end{lem}
\begin{proof}
The result can be proved by a minor modification of Lemma 1 of \cite{josanxiaand12}. The proof is presented in Appendix \ref{sec:proofassocpr} for completeness.
\end{proof}
\begin{figure*}[!t]
\setcounter{MYtempeqncnt}{\value{equation}}
\setcounter{equation}{18}
\begin{equation}\label{eq:pcov}
\pcov = \sum_{(i,j)\in \setopen{V}{}} 2\pi\dnstysrat{i}{j} \int_{0}^{\infty}\ndistnsc \exp\left(-\frac{\SINRthresh_{ij}}{\SNR_{ij}(y)} -\pi\left\{\sum_{k\in \set{V}_i} D_{ij}(k,\SINRthresh_{ij})\ndistnsc^{2/\nple{m}{k}}+ \sum_{\mk \in \setopen{V}{}}G_{ij}(m,k)\ndistnsc^{2/\nple{m}{k}}\right\} \right)\mathrm{d} \ndistnsc \, ,
\end{equation}
\setcounter{equation}{\value{MYtempeqncnt}}
\hrulefill
\vspace*{4pt}
\end{figure*}
The following two remarks provide  alternate interpretations of the association probability.
\begin{rem}
The  probability that a typical user is associated with the $i^{\mathrm{th}}$ RAT  is given by $\passoc_i = \sum_{j\in \setopen{V}{i}} \passoc_{ij}$.
This probability is also the average fraction of  the traffic offloaded, referred henceforth as \textit{traffic offload fraction},  to the $i^\mathrm{th}$ RAT.
\end{rem}
\begin{rem}\label{rem:assoc_area}
Using the ergodicity of the PPP, $\passoc_{ij}$ is the average fraction of the total area covered by the association regions of the APs of $(i,j)$.
\end{rem}
Based on  Remark \ref{rem:assoc_area}  we note that the  mean association area of a typical  AP of $(i,j)$ is $\frac{\passoc_{ij}}{\dnstysrat{i}{j}}$. Below we propose a linear scaling based approximation for association areas in HetNets, which matches this first moment. The results based on the area approximation are validated in Section \ref{sec:validation}.\\
{\textbf{Area Approximation}}: The area  $\area_{ij}$  of a   typical AP of the  $j^{\text{th}}$ tier of the  $i^{\text{th}}$ RAT can be approximated as
\begin{equation}\label{eq:area_approx}
\area_{ij} = \area\left(\frac{\dnstysrat{i}{j}}{\passoc_{ij}}\right),
\end{equation}
where $\area\left(y\right)$ is the area of a typical PV of density $y$ (a scale parameter).
\begin{rem}
The approximation is trivially exact for a single tier, single RAT scenario, i.e.,  for $\|\set{V}\|=1$.
\end{rem}
\begin{rem}
If $\metric{m}{k}\equiv \mathrm{T}$ and $\ple{m}{k} \equiv \ple{}{}$, then the approximation is exact. In this case, $\passoc_{ij} = \frac{\dnstysrat{i}{j}}{\sum_{(m,k)\in \setopen{V}{}}\dnstysrat{m}{k}}$ and
\begin{equation}
\area\left(\frac{\dnstysrat{i}{j}}{\passoc_{ij}}\right)= \area\left(\sum_{(m,k)\in \setopen{V}{}}\dnstysrat{m}{k}\right).
\end{equation}
With equal association weights and path loss coefficients, the HetNet model becomes the superposition of independent  PPPs, which is again a PPP with density equal to the sum of that of the constituents  and hence the resulting tessellation  is a PV. The right hand side of the above equation is equivalent to a typical association area of a PV  with density $\sum_{(m,k)\in \setopen{V}{}}\dnstysrat{m}{k}$.
\end{rem}
\begin{rem}
Using the  distribution proposed in \cite{Ferenc2007} for $\area(y)$, the distribution of $\area_{ij}$ is
\begin{equation}\label{eq:areadist}
f_{\area_{ij}}(\areans)= \frac{3.5^{3.5}}{\Gamma(3.5)}\frac{\dnstysrat{i}{j}}{\passoc_{ij}}\left(\frac{\dnstysrat{i}{j}}{\passoc_{ij}}\areans\right)^{2.5}\exp\left(-3.5\frac{\dnstysrat{i}{j}}{\passoc_{ij}}\areans\right),
\end{equation}
where $\Gamma(x) = \int_0^\infty \exp(-t)t^{x-1}\mathrm{d} t$ is the gamma function.
\end{rem}

To characterize the load at the tagged AP (AP serving the typical mobile user)  the implicit area biasing needs to be considered and the PGF of the number of  \textit{other} -- apart from the typical -- users ($\oload{i}{j}$) associated with the tagged AP needs to be characterized.
%%%%%%%%LEMMA%%%%%%%%%
\begin{lem}\label{lem:oloadpgf}
The PGF of the other users associated with the tagged AP of $(i,j)$ is
\begin{equation}
\pgf_{\oload{i}{j}}(z)= 3.5^{4.5}\left(3.5+\frac{\userdnsty\passoc_{ij}}{\dnstysrat{i}{j}}(1-z)\right)^{-4.5}.
\end{equation}
Furthermore, the moments of
$\oload{i}{j}$ are given by
\begin{equation}
\expect{\oload{i}{j}^n}=\sum_{k=1}^n\left(\frac{\userdnsty\passoc_{ij}}{\dnstysrat{i}{j}}\right)^k \stirling{n}{k}\expect{\area^{k+1}(1)},
\end{equation}
where $\stirling{n}{k}$ are Stirling numbers of the second kind\footnote{The notation of Stirling numbers given by $\stirling{n}{k}$ should not be confused with that of $\SINR$ coverage, $\pcov$.}.
\end{lem}
\begin{IEEEproof}
See Appendix \ref{sec:proofoloadpgf}.
\end{IEEEproof}
The moments of the typical association region of a PV of unit density can be computed numerically and are also available in  \cite{Gil62}.

\subsection{\texorpdfstring{$\SINR$}{SINR} Distribution}
The $\SINR$ of a typical user associated with an AP of $(i,j)$ located at $\ndistnsc$ is
\begin{equation}\label{eq:sinr}
\SINR_{ij}(\ndistnsc) = \frac{\power{i}{j}h_{\ndistnsc} \ndistnsc^{-\ple{i}{j}}}{\sum_{k \in \set{V}_i} I_{ik} + \noisepower{i}},
\end{equation}
where $h_\ndistnsc$ is the channel gain from the tagged  AP located at a distance  $\ndistnsc$, $I_{ik}$ denotes the interference from the APs of RAT $i$ in the tier $k$. The set of APs contributing to interference are from $\tierPPP{i}{k} \bigcup \PPPcrat{i}{k} \setminus o\,\forall k \in \set{V}_i$ %\,\forall\, k,k' \in \set{V}_i$
, where  $o$ denotes the tagged AP from $(i,j)$. Thus
\begin{equation}
I_{ik} = \power{i}{k}\sum_{x\in \tierPPP{i}{k} \setminus o}  h_{x} x^{-\ple{i}{k}} + \power{i}{k}\sum_{x'\in \PPPcrat{i}{k}}  h_{x'} {x'}^{-\ple{i}{k}}.
\end{equation}
For a typical user, when associated with $(i,j)$, the probability that the received $\SINR$ is greater than a threshold  $\SINRthresh_{ij}$, or $\SINR$ coverage, is
\begin{equation}
\pcov_{ij}(\SINRthresh_{ij}) \triangleq \cexpect{\pr\{\SINR_{ij}(\ndistnsc)> \SINRthresh_{ij}\}}{\ndistnsc},
\end{equation}
and the overall $\SINR$ coverage is
\begin{equation}
\pcov = \sum_{(i,j)\in \setopen{V}{}}\pcov_{ij} (\SINRthresh_{ij})\passoc_{ij}.
\end{equation}

\begin{figure*}[!t]
\setcounter{MYtempeqncnt}{\value{equation}}
\setcounter{equation}{19}
\begin{align}\label{eq:ncov}
\psys =&\sum_{(i,j)\in \setopen{V}{}} 2\pi\dnstysrat{i}{j}  \sum_{n\ge 0}\frac{3.5^{3.5}}{n!}\frac{\Gamma(n+4.5)}{\Gamma(3.5)}\left(\frac{\userdnsty\passoc_{ij}}{\dnstysrat{i}{j}}\right)^n\left(3.5 + \frac{\userdnsty\passoc_{ij}}{\dnstysrat{i}{j}}\right)^{-(n+4.5)}\nonumber\\
&\times\int_{0}^{\infty}\ndistnsc \exp\left(-\frac{\uRATEthresh(\nRATEthresh_{ij}(n+1))}{\SNR_{ij}(y)} -\pi\left\{\sum_{k\in \set{V}_i} D_{ij}(k,\uRATEthresh(\nRATEthresh_{ij}(n+1)))\ndistnsc^{2/\nple{m}{k}}+ \sum_{\mk \in \setopen{V}{}}G_{ij} (m,k)\ndistnsc^{2/\nple{m}{k}}\right\} \right)\mathrm{d} \ndistnsc \, ,\end{align}
\setcounter{equation}{\value{MYtempeqncnt}}
\hrulefill
\vspace*{4pt}
\end{figure*}

\begin{figure*}[!t]
\setcounter{MYtempeqncnt}{\value{equation}}
\setcounter{equation}{26}
\begin{equation}\label{eq:rcovmeanload}
\bar{\psys} =\sum_{(i,j)\in \setopen{V}{}} 2\pi\dnstysrat{i}{j}\int_{0}^{\infty}\ndistnsc \exp\Bigg(-\frac{\uRATEthresh(\nRATEthresh_{ij}\avload{i}{j})}{\SNR_{ij}(y)} -\pi\bigg\{\sum_{k\in \set{V}_i} D_{ij}(k,\uRATEthresh(\nRATEthresh_{ij}\avload{i}{j}))\ndistnsc^{2/\nple{m}{k}}+ \sum_{\mk \in \setopen{V}{}}G_{ij} (m,k)\ndistnsc^{2/\nple{m}{k}}\bigg\} \Bigg)\mathrm{d} \ndistnsc \, ,
\end{equation}
\setcounter{equation}{\value{MYtempeqncnt}}
\hrulefill
\vspace*{4pt}
\end{figure*}
Interestingly, the distance of a typical user to the tagged  AP in $(i,j)$, $\NDISTc{i}{j}$,  is not only influenced by $\PPPsrat{i}{j}$ but also by $\PPPsrat{m}{k}\,\, \forall \mk \in \setopen{V}{}$ , as APs of other open access classes also compete to become the serving AP. The distribution of this  distance  is given by the following lemma.
\begin{lem}\label{lem:ndist}
The probability distribution function (PDF), $f_{\NDISTc{i}{j}}(\ndistnsc)$, of the distance $\NDISTc{i}{j}$ between a typical user and the tagged AP of $(i,j)$ is
\begin{equation}
f_{\NDISTc{i}{j}}(\ndistnsc) = \frac{2\pi\dnstysrat{i}{j}}{\passoc_{ij}} \ndistnsc\exp\left\{-\pi\sum_{\mk \in \setopen{V}{}}G_{ij}(m,k)\ndistnsc^{2/\nple{m}{k}}\right\}.
\end{equation}
\end{lem}
\begin{proof}
See Appendix \ref{sec:proofndist}.
\end{proof}
The following lemma gives the $\SINR$ CCDF/coverage  over the entire network.
\addtocounter{equation}{1}
\begin{lem}\label{lem:pcov}
The $\SINR$  coverage  of a typical user is given by (\ref{eq:pcov}) (at the top of page)
where  \small
\begin{equation*}
 D_{ij}(k,\SINRthresh_{ij})\\ = \npower{i}{k}^{2/\ple{i}{k}}\left\{\tierdnsty{i}{k}\Z\left(\SINRthresh_{ij},\ple{i}{k},\nmetric{i}{k}\npower{i}{k}^{-1}\right) + \dnstycrat{i}{k}\Z(\SINRthresh_{ij},\ple{i}{k},0)\right\} ,\end{equation*}\normalsize
\[G_{ij}(m,k)= \dnstysrat{m}{k} \nmetric{m}{k}^{2/\ple{m}{k}},\,\,\, \Z(a,b,c)= a^{2/b}\int_{(\frac{c}{a})^{2/b}}^\infty \frac{\mathrm{d} u}{ 1+ u^{b/2}}\,\,,\]
$ \mathrm{ and }\,\, \SNR_{ij}(y)= \frac{\power{i}{j}\ndistnsc^{-\ple{i}{j}}}{\noisepower{i}}.$
\end{lem}
\begin{IEEEproof}
See Appendix \ref{sec:proofpcov}.
\end{IEEEproof}
The result in Lemma \ref{lem:pcov} is for the most general case and involves a single numerical integration  along with a lookup table for $\Z$.  Lemma \ref{lem:pcov} reduces to the earlier derived $\SINR$ coverage expressions in \cite{andganbac11} for $M=K=1$ (single tier, single RAT) and those in  \cite{josanxiaand12} for   $M=1$ (single RAT, multiple tiers).
\subsection{Main Result}
Having characterized the distribution of load and $\SINR$, we now derive the rate distribution over the whole network.
%%%%%%%%THEOREM%%%%%%%%%%%%%%%%%%%
\addtocounter{equation}{1}
\begin{thm}\label{thm:rcov}
The rate coverage  of a randomly located mobile user in the general HetNet setting of Section  \ref{sec:sysmodel}  is given by (\ref{eq:ncov}) (at the top of  page) where $\RATEthresh_{ij}$ is the rate threshold for $(i,j)$,    $\nRATEthresh_{ij}\triangleq\RATEthresh_{ij}/\res{i}{j}$, and $\uRATEthresh(x) \triangleq 2^{x}-1$.
\end{thm}
\begin{IEEEproof}
Using (\ref{eq:ratemodel}), the probability that the rate requirement of a user associated with $(i,j)$ is met is
\begin{align}
\pr(\rate{i}{j}> \RATEthresh_{ij})&= \pr\left(\frac{\res{i}{j}}{\load{i}{j}}\log(1+\SINR_{ij})>\RATEthresh_{ij}\right)\nonumber\\
&=\pr (\SINR_{ij} > 2^{\RATEthresh_{ij} \load{i}{j}/\res{i}{j}} -1)\\
& = \cexpect{\pcov_{ij}\left(\uRATEthresh(\nRATEthresh_{ij}\load{i}{j})\right)}{\load{i}{j}},\label{eq:rcovij}
\end{align}
where
$\uRATEthresh(\nRATEthresh_{ij}\load{i}{j}) = 2^{\RATEthresh_{ij}\load{i}{j}/\res{i}{j}}-1 $
and $\load{i}{j} = 1+ \oload{i}{j},$ i.e., the load at the tagged AP  equals the typical user plus the \textit{other} users.
Using Lemma  \ref{lem:oloadpgf}, (\ref{eq:rcovij}) is simplified as
\begin{align}
&\cexpect{\pcov_{ij}\left(\uRATEthresh(\nRATEthresh_{ij}\load{i}{j})\right)}{\load{i}{j}}\nonumber\\
&= \sum_{n\ge 0}\pr(\oload{i}{j}=n)\pcov_{ij}\left(\uRATEthresh(\nRATEthresh_{ij}(n+1))\right)\\
&= \sum_{n\ge 0}\frac{3.5^{3.5}}{n!}\frac{\Gamma(n+4.5)}{\Gamma(3.5)}\left(\frac{\userdnsty\passoc_{ij}}{\dnstysrat{i}{j}}\right)^n\nonumber\\
&\phantom{"="}\times\left(3.5 + \frac{\userdnsty\passoc_{ij}}{\dnstysrat{i}{j}}\right)^{-(n+4.5)}\pcov_{ij}\left(\uRATEthresh(\nRATEthresh_{ij}(n+1))\right).
\end{align}
Using the law of total probability, the rate coverage is
\small
\begin{align}
\psys & =  \sum_{(i,j) \in \setopen{V}{}} \passoc_{ij} \pr(\rate{i}{j} > \RATEthresh_{ij})=\sum_{(i,j) \in \setopen{V}{}}\passoc_{ij}\sum_{n\ge 0}\frac{3.5^{3.5}}{n!}\frac{\Gamma(n+4.5)}{\Gamma(3.5)}\nonumber\\
&\times\left(\frac{\userdnsty\passoc_{ij}}{\dnstysrat{i}{j}}\right)^n\left(3.5 + \frac{\userdnsty\passoc_{ij}}{\dnstysrat{i}{j}}\right)^{-(n+4.5)}\pcov_{ij}\left(\uRATEthresh(\nRATEthresh_{ij}(n+1))\right).
\end{align}
\normalsize
Using Lemma \ref{lem:pcov} in the above equation gives the desired result.
\end{IEEEproof}
The rate distribution expression for the most general setting requires a single numerical integral and use of lookup tables for $\Z$ and $\Gamma$.  Since both the terms $\pr(\load{i}{j}=n)$ and $\pcov_{ij}\left(\uRATEthresh(n)\right)$ decay rapidly for large $n$, the summation over $n$ in Theorem \ref{thm:rcov} can be accurately approximated as a finite summation  to a sufficiently large value, ${N}_\text{max}$. We found  ${N}_\text{max} = 4\userdnsty$ to be sufficient for  results presented in Section \ref{sec:validation}.
%%%%Special cases%%%%%%%%%%%%%%%%%%%%%%

\subsection{Mean Load Approximation}
The rate coverage expression can be further simplified (sacrificing accuracy) if the load at each AP of  $(i,j)$ is assumed equal to its mean.
\addtocounter{equation}{1}
\begin{cor}\label{cor:rcovmeanload}
Rate coverage with the mean load approximation is given by (\ref{eq:rcovmeanload}) (at the top of  page),
where
\begin{equation*}
\avload{i}{j}= \expect{\load{i}{j}} = 1+ \frac{1.28\userdnsty\passoc_{ij}}{\dnstysrat{i}{j}}.
\end{equation*}
\end{cor}
\begin{IEEEproof}
Lemma \ref{lem:oloadpgf} gives the first moment of load as
$\expect{\load{i}{j}} = 1+ \expect{\oload{i}{j}}= 1+ \frac{\userdnsty\passoc_{ij}}{\dnstysrat{i}{j}}\expect{\area^2(1)}$ where
 $\expect{\area^2(1)}=1.28$ \cite{Gil62}. Using an approximation for (\ref{eq:rcovij}) with
$\cexpect{\pcov_{ij}\left(\uRATEthresh(\nRATEthresh_{ij}\load{i}{j})\right)}{\load{i}{j}} \approx
\pcov_{ij}\left(\uRATEthresh(\nRATEthresh_{ij}\expect{\load{i}{j}})\right)$, the simplified rate coverage expression is obtained.
\end{IEEEproof}
The mean load approximation above simplifies the rate coverage expression by eliminating the summation over $n$. The numerical integral can also be eliminated in certain plausible scenarios given in the following corollary.
\begin{cor}\label{cor:simnwcov}
In interference limited scenarios ($\noisepower{}\to 0$) with mean load approximation and with same path loss exponents (${\nple{m}{k}} \equiv 1$), the rate coverage is
\small
\begin{equation}
\bar{\psys}= \sum_{(i,j) \in \setopen{V}{}} \frac{\dnstysrat{i}{j}}{ \sum_{k\in \set{V}_i}D_{ij}(k,\uRATEthresh(\nRATEthresh_{ij}\avload{i}{j}))  + \sum_{\mk \in \setopen{V}{}}G_{ij}(m,k)}.
\end{equation}
\normalsize
\end{cor}

In the above analysis, rate distribution is presented as a function of association weights. So, in principle, it is possible to find the optimal association weights and hence the optimal fraction of  traffic to be offloaded to each RAT so as to maximize the rate coverage. This aspect is studied in a special case of a two-RAT network in Section \ref{sec:opt_offload}.

\subsection{Validation}\label{sec:validation}
In this section, the emphasis is on validating the area and mean load approximations proposed  for rate coverage and on validating the PPP as a  suitable AP location model.  In all the simulation results, we consider a square window of $20\times20$ km$^2$. The AP locations are drawn from  a PPP or a real deployment or a square grid  depending upon the scenario that is being simulated. The typical user is assumed to be located at the origin. The serving AP for this user (tagged AP) is determined by (\ref{eq:association}). The received $\SINR$ can now be evaluated as being the ratio of the power received from the serving AP and the sum of the powers received from the rest of the APs as given in (\ref{eq:sinr}). The rest of the users are assumed to form a realization of an independent PPP. The serving AP of each user is again determined by (\ref{eq:association}), which provides the total load on the tagged AP in terms of the number of users it is serving. The rate of the typical user is then computed according to (\ref{eq:ratemodel}). In each Monte-Carlo trial, the user locations, the base station locations, and the channel gains are independently generated. The rate distribution is obtained by simulating  $10^5$ Monte-Carlo trials.

In the discussion that follows we use a specific form of the association weight as $\metric{m}{k} =  \power{m}{k}\bias{m}{k}$ corresponding to the biased received power based  association \cite{qcom_hetnet_wmag}, where $\bias{m}{k}$ is the association bias for $(m,k)$.  The effective resources  at an AP are assumed to be uniformly $\res{m}{k}\equiv10$ MHz and equal  rate thresholds are assumed for all classes. Thermal noise is ignored. Also, without any loss of generality the bias of $(1,1)$ is normalized to 1, or $\bias{1}{1}=0$ dB.

\subsubsection{Analysis}
Our goal here is to validate   the area approximation  and the mean load approximation (Theorem \ref{thm:rcov} and  Corollary \ref{cor:rcovmeanload}, respectively) in the context of rate coverage. A scenario with two-RATs, one with a single open access tier and the other with two tiers -- one open and one closed access -- is considered first. In this case, $\set{V}=\{(1,1);(2,3);(2,3^{'})\}$, $\dnstysrat{1}{1}=1$ BS/km$^2$, $\dnstysrat{2}{3}=\dnstycrat{2}{3}=10$ BS/km$^2$, $\userdnsty = 50$ users/km$^2$, $\ple{1}{1}=3.5$, and  $\ple{1}{3}=4$. Fig. \ref{fig:sim_offload_twotier_closedaccess} shows the rate distribution obtained through simulation and that from Theorem \ref{thm:rcov} and Corollary \ref{cor:rcovmeanload} for two values of association biases. Fig. \ref{fig:sim_offload_tworat_twotier} shows the the rate distribution in a two-RAT three-tier setting  with $\set{V}=\{(1,1);(1,2);(2,2);(2,3)\}$, $\dnstysrat{1}{1}=1$ BS/km$^2$, $\dnstysrat{1}{2}=\dnstysrat{2}{2}=5$ BS/km$^2$, $\dnstysrat{2}{3} = 10$ BS/km$^2$, $\userdnsty = 50$ users/km$^2$, $\ple{1}{1}=3.5$, $\ple{1}{2}=3.8$, and  $\ple{1}{3}=4$ for two values of association bias of $(2,3)$. In both cases, $\bias{1}{2} = \bias{2}{2} = 5$ dB.

As it can be observed from both the plots, the analytic distributions obtained from  Theorem \ref{thm:rcov} and Corollary \ref{cor:rcovmeanload} are in quite good agreement  with the simulated  one and thus validate the analysis.  See \cite{SinDhiAndICC13} for validation of  a three-RAT scenario.
\subsubsection{Spatial Location Model}
To simulate a realistic spatial location model for a two-RAT scenario, the cellular BS location data of a major metropolitan city used in \cite{andganbac11} is overlaid with that of an actual WiFi deployment \cite{googlemvwifi}. Along with the  PPP,
a square grid based location model in which the APs for both the RATs are located in a square lattice (with different densities) is also used in the following comparison.
Denoting the macro tier as $(1,1)$ and  WiFi APs as $(2,3)$, $\set{\set{V}}=\{(1,1);(2,3)\}$ in this setup. The superposition is done such that $\dnstysrat{2}{3}=10\dnstysrat{1}{1}$. Fig. \ref{fig:mvdalls_ratecovcombtwc} shows the rate distribution of a typical user obtained from the real data along with that of a square grid based model and that from a PPP,  Theorem \ref{thm:rcov}, for three cases.  As evident from  the plot, Theorem \ref{thm:rcov} is quite accurate in the context of rate distribution with regards to the actual location data.
\begin{figure}
	\centering
		\includegraphics[width=\columnwidth]{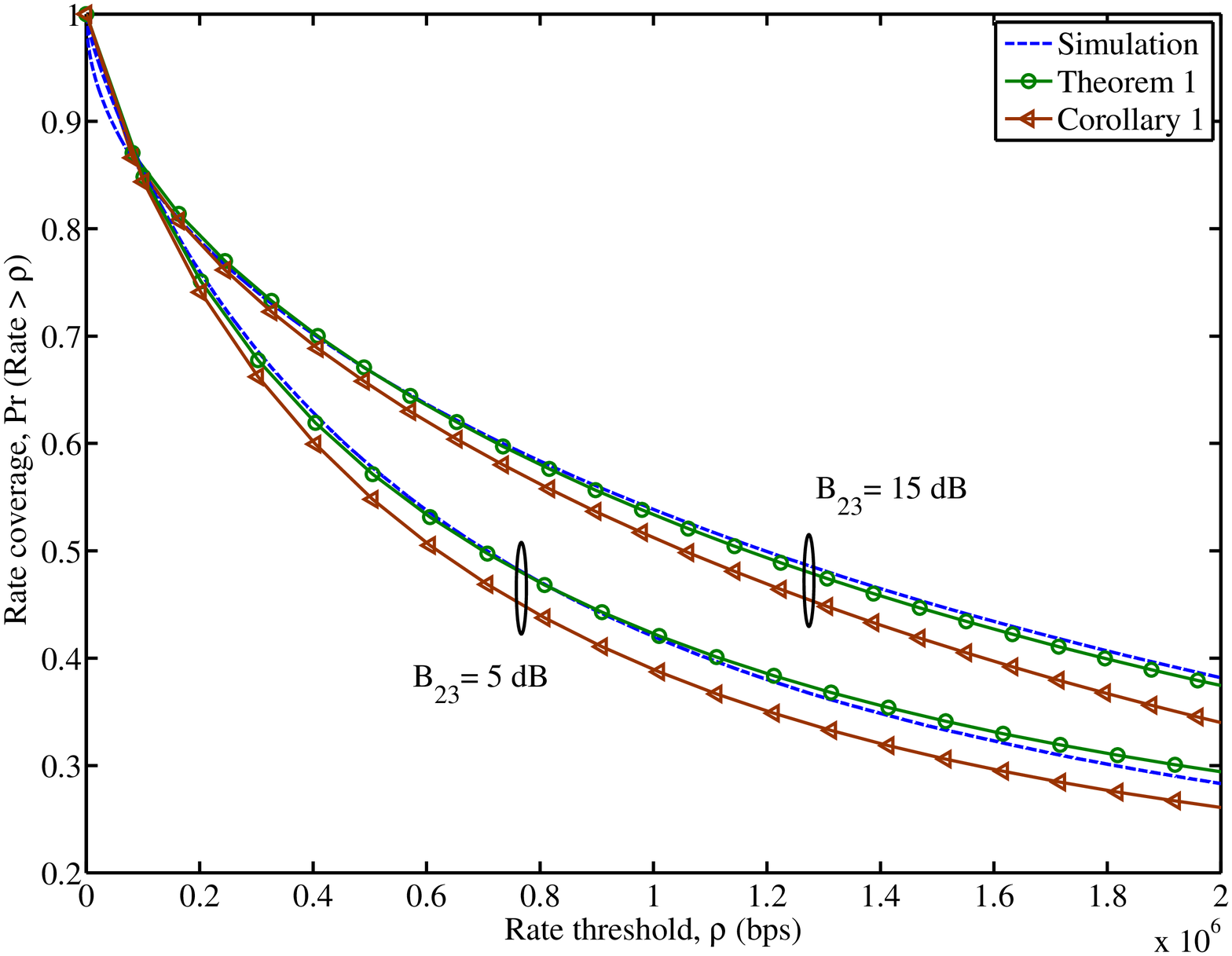}
		\caption{Comparison of rate distribution obtained from simulation, Theorem 1, and Corollary 1 for $\dnstysrat{2}{3}=\dnstycrat{2}{3}=10\dnstysrat{1}{1}$, $\ple{1}{1}=3.5 $, and  $\ple{2}{3} = 4$.}
	\label{fig:sim_offload_twotier_closedaccess}
\end{figure}
\begin{figure}
	\centering
		\includegraphics[width=\columnwidth]{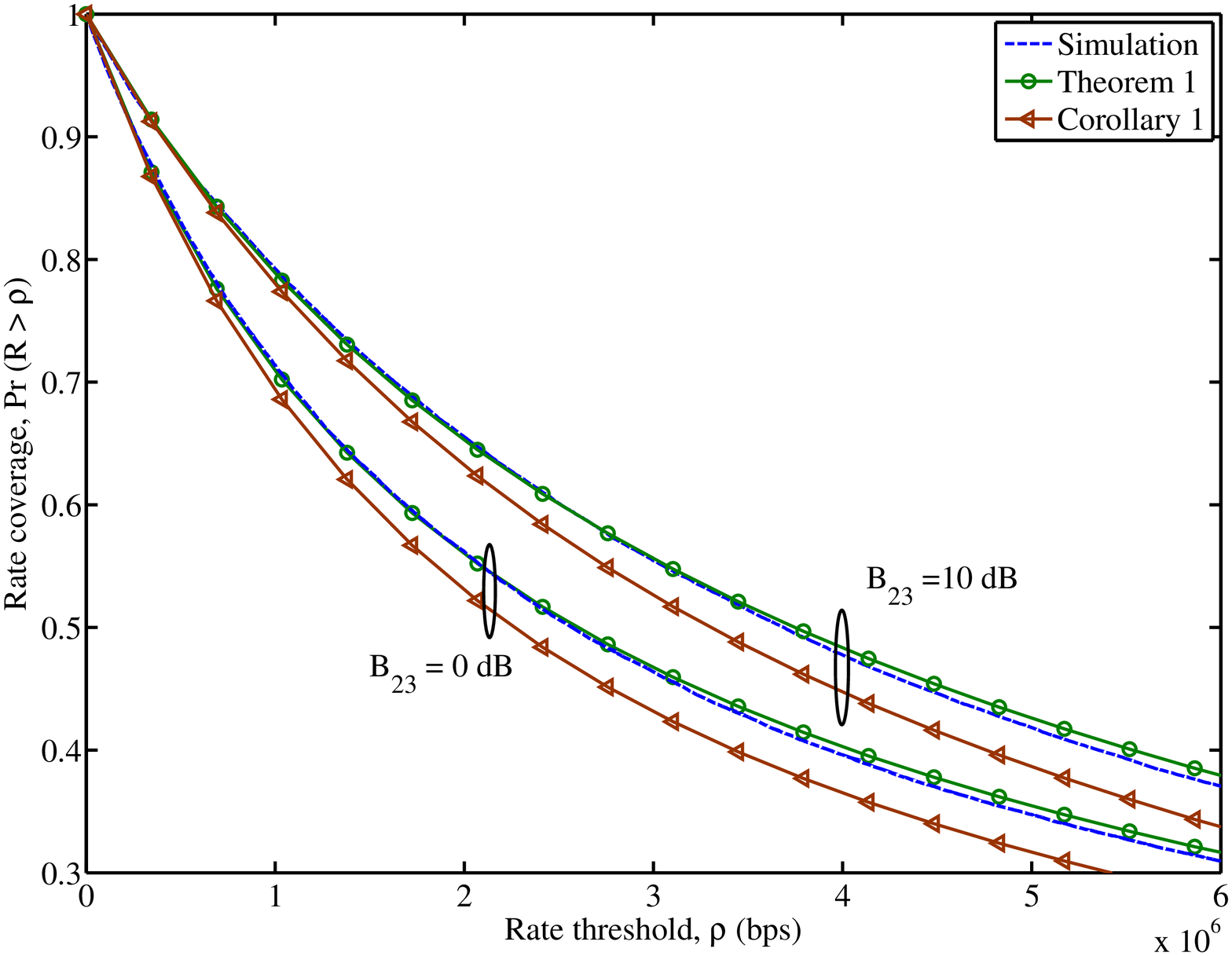}
		\caption{Comparison of rate distribution obtained from simulation, Theorem 1, and Corollary 1 for $\dnstysrat{1}{2}=\dnstysrat{2}{2}=5\dnstysrat{1}{1}$, $\dnstysrat{2}{3}=10\dnstysrat{1}{1}$, $\ple{1}{1}=3.5 $, $\ple{1}{2}=3.8$, and  $\ple{2}{3} = 4$.}
	\label{fig:sim_offload_tworat_twotier}
\end{figure}
\begin{figure}
	\centering
		\includegraphics[width=\columnwidth]{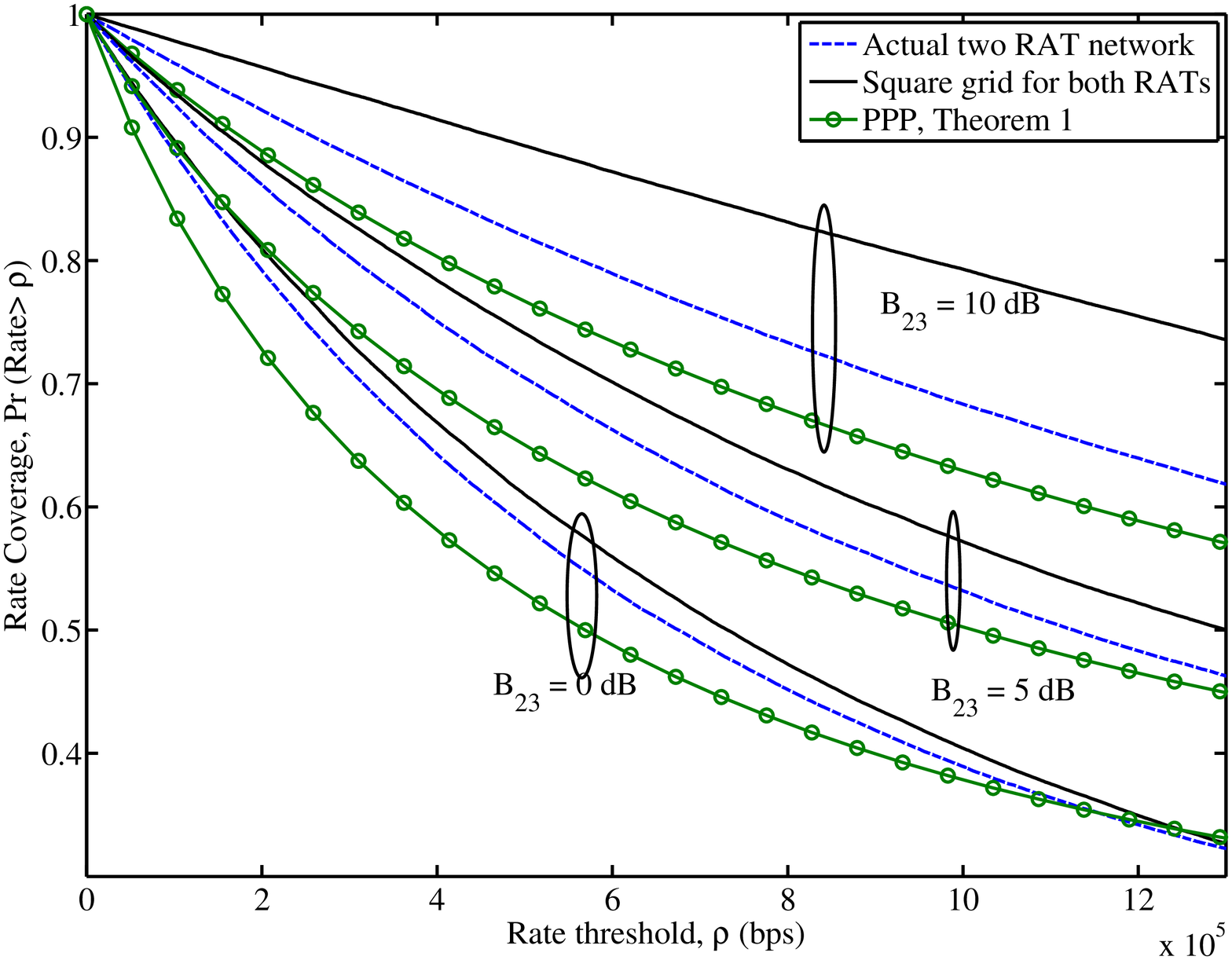}
		\caption{Rate distribution comparison for the three spatial location models: real, grid, and PPP for a two-RAT setting with $\dnstysrat{2}{3}=10\dnstysrat{1}{1}$ and $\ple{1}{1}=\ple{2}{3}=4$}
	\label{fig:mvdalls_ratecovcombtwc}
\end{figure}
\section{Design of Optimal offload}\label{sec:opt_offload}
In this section, we consider the design of optimal offloading   under a specific form of the association weight as  $\metric{m}{k}=\power{m}{k}\bias{m}{k}$.
For general settings, the optimum association biases $\{\bias{m}{k}\}$  for $\SINR$ and rate coverage can be found using the derived expressions of  Lemma 5 and Theorem 1 respectively. As discussed in  Section III-E, simplified expression of Corollary 1 can also be used for rate coverage.
We consider below a two-RAT single tier scenario with $q^\text{th}$ tier of RAT-1 overlaid with $r^\text{th}$ tier of RAT-2, i.e., $\set{V}=\{(1,q);(2,r)\}$. Optimal association  bias and optimal traffic offload  fraction is  investigated here in the  context of both the $\SIR$  coverage (i.e., neglecting noise)  and  rate coverage.

\subsection{Offloading for Optimal \texorpdfstring{$\SIR$}{SIR} Coverage}
\begin{prop}\label{prop:optbias}
Ignoring thermal noise (interference limited scenario, $\noisepower{} \to 0$), assuming equal path loss coefficients ($\nple{m}{k}\equiv 1$),   the value of association  bias  $\frac{\bias{2}{r}}{\bias{1}{q}}$  maximizing $\SIR$ coverage is
\begin{equation}
b_\mathrm{opt}=\frac{\power{1}{q}}{\power{2}{r}}\left(\frac{\Z(\SINRthresh_{1q},\alpha,1)}
{a\Z(\SINRthresh_{2r},\alpha,1)}\right)^{\alpha/2},
\end{equation}
where  $\dnstysrat{2}{r} = a\dnstysrat{1}{q}$ and  the corresponding  optimum traffic  offload fraction to RAT-2 is
\begin{equation}
\passoc_2 = \frac{\Z(\SINRthresh_{1q},\alpha,1)}{\Z(\SINRthresh_{2r},\alpha,1)+\Z(\SINRthresh_{1q},\alpha,1)}.
\end{equation}
The corresponding $\SIR$ coverage is \begin{equation}
\frac{\Z(\SINRthresh_{2r},\alpha,1)+\Z(\SINRthresh_{1q},\alpha,1)}{\Z(\SINRthresh_{2r},\alpha,1)+\Z(\SINRthresh_{1q},\alpha,1) + \Z(\SINRthresh_{2r},\alpha,1)\Z(\SINRthresh_{1q},\alpha,1)}.
\end{equation}
\end{prop}
\begin{proof}
See Appendix \ref{sec:proofoptbias}.
\end{proof}
The following observations can be made from the above Proposition:
\begin{itemize}
\item The optimal bias for $\SIR$ coverage  is inversely proportional to the density and transmit power of the corresponding RAT. This is because the  denser the second RAT and the higher the transmit power of the corresponding APs, the higher the interference experienced by offloaded users leading to a decrease in the optimal bias. Also, with increased density and power,  lesser bias is required to offload the same fraction of traffic.
\item The optimal fraction of traffic/user population to be offloaded to either  RAT for maximizing $\SIR$ coverage is \textit{independent} of the density and power and is solely dependent on $\SIR$ thresholds. The higher the RAT-1 threshold, $\SINRthresh_{1q}$, compared  to that of  RAT-2 threshold, $\SINRthresh_{2r}$, the more percentage of traffic is offloaded to RAT-2 as $\Z$ is a monotonically increasing function of $\SINRthresh$.  In fact, if $\SINRthresh_{1q}= \SINRthresh_{2r} $, offloading \textit{half} of the user population  maximizes $\SIR$ coverage.
\end{itemize}

\subsection{Offloading for Optimal Rate Coverage}
For the design of  optimal offloading  for rate coverage, the mean load approximation (Corollary \ref{cor:rcovmeanload}) is used.
\begin{prop}\label{prop:optbiasncov}
Ignoring thermal noise (interference limited scenario, $\noisepower{} \to 0$), assuming equal path loss coefficients ($\nple{m}{k}\equiv 1$),   the value of association  bias  $\frac{\bias{2}{r}}{\bias{1}{q}}$  maximizing rate coverage is
\small
\begin{align}
b_\mathrm{opt} & = \arg \max_{b}\Bigg\{\left(\Z(\uRATEthresh_{1q}(\nRATEthresh_{1q}\avload{1}{q}),\alpha,1)+1 +a(\npower{2}{r}b)^{2/\alpha}\right)^{-1}\nonumber\\
& + \left(\Z(\uRATEthresh_{2r}(\nRATEthresh_{2r}\avload{2}{r}),\alpha,1)+ 1 +\frac{1}{a(\npower{2}{r}b)^{2/\alpha}}\right)^{-1}\Bigg\},
\end{align}
\normalsize
where $a =\dnstysrat{2}{r}/\dnstysrat{1}{q}$ and $b=\bias{2}{r}/\bias{1}{q}$.
\end{prop}
\begin{IEEEproof}
The optimum association bias  can be found by maximizing the expression obtained from Corollary \ref{cor:simnwcov}  using $\set{V}=\{(1,q);(2,r)\}$, $\dnstysrat{2}{r} = a\dnstysrat{1}{q}$, and $\bias{2}{r} = b\bias{1}{q}$.
\end{IEEEproof}
Unfortunately, a closed form expression for the optimal bias is not possible in this case, as the load (and hence the threshold) is dependent on the association bias $b$. However, the optimal  association bias, $b_\mathrm{opt}$, for the rate coverage can be found out through a  linear search using the above Proposition. In a general setting, the computational complexity of finding the optimal biases, however, increases with the number of classes of APs in the network as the dimension of the problem increases. While the exact computational complexity depends upon the choice of optimization algorithm, the proposed analytical approach is clearly less complex than exhaustive simulations by virtue of the easily computable rate coverage expression.

The analysis in this section shows  that for a two-RAT scenario, $\SIR$ coverage and rate  coverage exhibit considerably different behavior. The optimal traffic offload fraction for $\SIR$ coverage is independent of the density whereas for rate coverage it is expected to increase because of the decreasing load per AP for the second RAT.  For a fixed bias, rate coverage always increases with density, however for a fixed density there is always  an optimal traffic offload fraction. These insights might be known to practicing wireless system engineers but here a theoretical analysis makes the observations rigorous.

%%%%%%%%%%%%%%PERFORMANCE ANALYSIS%%%%%%%%%%%
\section{Results and Discussion}\label{sec:results}
In this section we primarily consider a setting of macro tier of RAT-1 overlaid with a low power tier of RAT-2, i.e., $\set{V}=\{(1,1);(2,3)\}$. This setting is similar to the widespread use of  WiFi APs to offload the macro cell traffic. In particular, the effect of association bias and traffic offload fraction on $\SIR$ and rate coverage is investigated. Thermal noise is ignored in the following results.
\subsection{\texorpdfstring{$\SIR$}{SIR} coverage}
The variation  of $\SIR$ coverage  with the density of RAT-2 APs for different values of  association bias is shown in Fig. \ref{fig:pcov_density}. The path loss exponent used is $\ple{m}{k}\equiv 3.5$ and the respective $\SIR$ thresholds are $\SINRthresh_{11} = 3$ dB and $\SINRthresh_{23} =6$ dB.  It is clear that for any fixed value of association bias, $\pcov$ is sub-optimal for all values of densities except for the bias value satisfying Proposition \ref{prop:optbias}. Also shown is the optimum $\SIR$ coverage (Proposition \ref{prop:optbias}), which is invariant to the density of APs.

Variation of $\SIR$ coverage  with the association bias is shown in Fig. \ref{fig:pcov_bias} for different densities of RAT-2  APs.  As shown, increasing density of RAT-2 APs decreases the optimal offloading bias. This is due to the corresponding increase in the interference for offloaded users in RAT-2. This insight will also be useful in rate coverage analysis. Again, at all values of association bias, $\pcov$ is sub-optimal for all density values except for the optimum density,  $\lambda_\mathrm{opt} = \left(\frac{\power{1}{q}}{\power{2}{r}\nbias{2}{r}}\right)^{2/\alpha}\frac{\Z(\SINRthresh_{1q},\alpha,1)}
{\Z(\SINRthresh_{2r},\alpha,1)}$.
\subsection{Rate Coverage}
The variation  of rate coverage with the density of RAT-2  APs for different values of association  bias is shown  in Fig. \ref{fig:rcov_density} and  the  variation   with the association bias is shown in Fig. \ref{fig:rcov_bias} for different  densities of RAT-2  APs. In these results, the user density $\userdnsty= 200$ users/km$^2$, the rate threshold $\RATEthresh_{mk}\equiv 256$ Kbps, the effective bandwidth $\res{m}{k} \equiv 10$ MHz, and the path loss exponent  is $\ple{m}{k} \equiv 3.5$. As expected, rate coverage increases  with increasing AP density because of the decrease in load at each AP. The optimum association bias for rate  coverage is obtained by a linear search as in Proposition \ref{prop:optbiasncov}.  For all values of association bias, $\psys$ is sub-optimal except for the one given in Proposition \ref{prop:optbiasncov}.
Fig. \ref{fig:rcov_fiveile} shows the effect of association bias on the $5^{\mathrm{th}}$ percentile rate $\RATEthresh_{95}$ with $\psys|_{\RATEthresh_{95}}=0.95$ (i.e., $95\%$ of the user population receives a rate greater than $\RATEthresh_{95}$)
for different  densities of RAT-2  APs. Comparing Fig. \ref{fig:rcov_bias} and Fig. \ref{fig:rcov_fiveile}, it can be seen that the optimal bias is agnostic to rate thresholds. This leads to the design insight that for given network  parameters  re-optimization is not needed for different rate thresholds. The developed analysis can also be used to find optimal biases for a more general setting. Fig. \ref{fig:rcov_fiveile_tworattwotier} shows the $5^{\mathrm{th}}$ percentile rate for a setting with $\set{V}=\{(1,1);(1,2);(2,2);(2,3)\}$, $\dnstysrat{1}{1}=1$ BS/km$^2$, $\dnstysrat{1}{2}=\dnstysrat{1}{2}=5$ BS/km$^2$,  $\bias{1}{2}=\bias{2}{2}=5$ dB as a function of association bias of $(2,3)$.
It can be seen that the choice of association biases  can heavily influence rate coverage.

A common observation  in  Fig. \ref{fig:rcov_bias}-\ref{fig:rcov_fiveile_tworattwotier} is the decrease in the optimal offloading bias with the increase in  density of APs of the corresponding  RAT. This can be explained by the earlier  insight of decreasing  optimal bias for  $\SIR$ coverage with increasing density. However, in contrast to the trend  in $\SIR$ coverage,  the optimum traffic offload fraction increases with increasing density as the corresponding load at each AP of second RAT decreases. These trends are further highlighted in Fig. \ref{fig:bias_trend} for the following scenarios:
\begin{itemize}
\item Case 1:  $\res{1}{1}= 15$ MHz, $\res{2}{3}= 5$ MHz, $\RATEthresh_{11} = 256$ Kbps, and $\RATEthresh_{23} = 512$ Kbps.
\item Case 2:  $\res{1}{1}= 5$ MHz, $\res{2}{3}= 15$ MHz, $\RATEthresh_{11} = 512$ Kbps, and $\RATEthresh_{23} = 256$ Kbps.
\end{itemize}
It can be seen  that apart from the effect of deployment density, optimum choice of association  bias and traffic offload fraction also depends  on the ratio of rate threshold ($\RATEthresh_{ij}$) to the bandwidth ($\res{i}{j}$), or $\nRATEthresh_{ij}$. In particular, larger the ratio of the available resources to the rate threshold  more is the tendency to be offloaded to the corresponding RAT.

\begin{figure}
	\centering
		\includegraphics[width=\columnwidth]{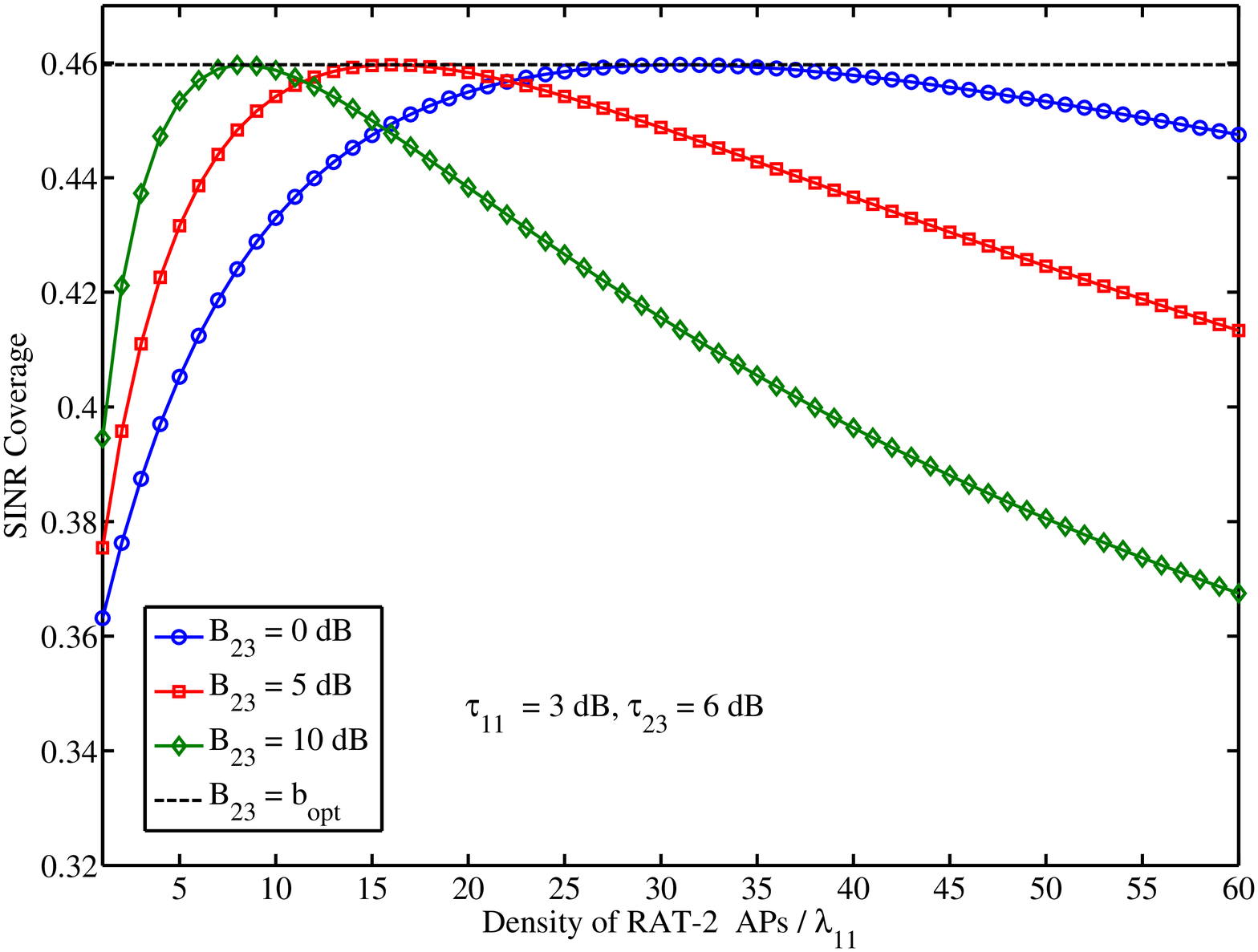}
		\caption{Effect of density of RAT-2  APs on $\SINR$ coverage.}
	\label{fig:pcov_density}
\end{figure}
\begin{figure}
	\centering
		\includegraphics[width=\columnwidth]{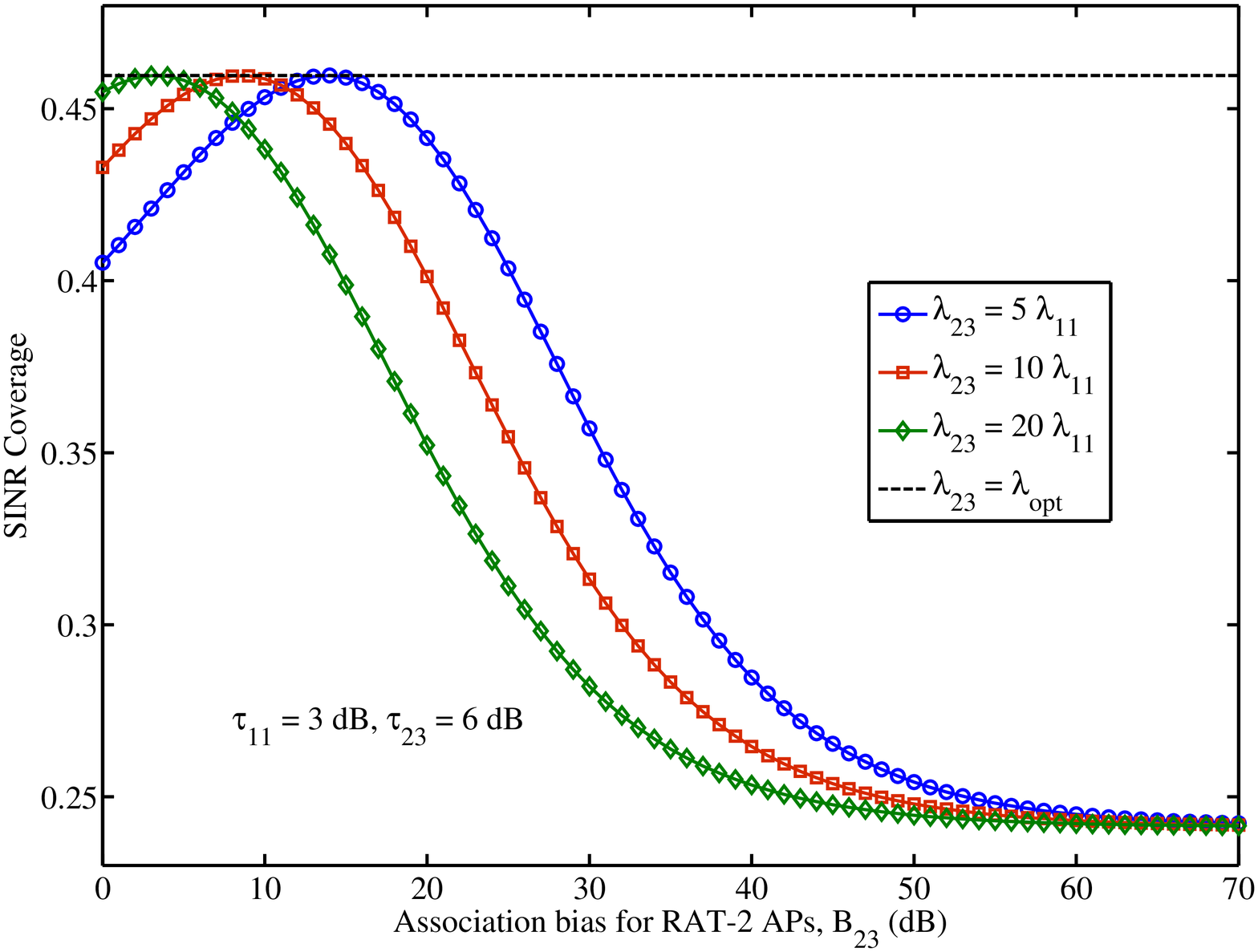}
		\caption{Effect of association bias for RAT-2 APs on $\SINR$ coverage.}
	\label{fig:pcov_bias}
\end{figure}
\begin{figure}
	\centering
		\includegraphics[width=\columnwidth]{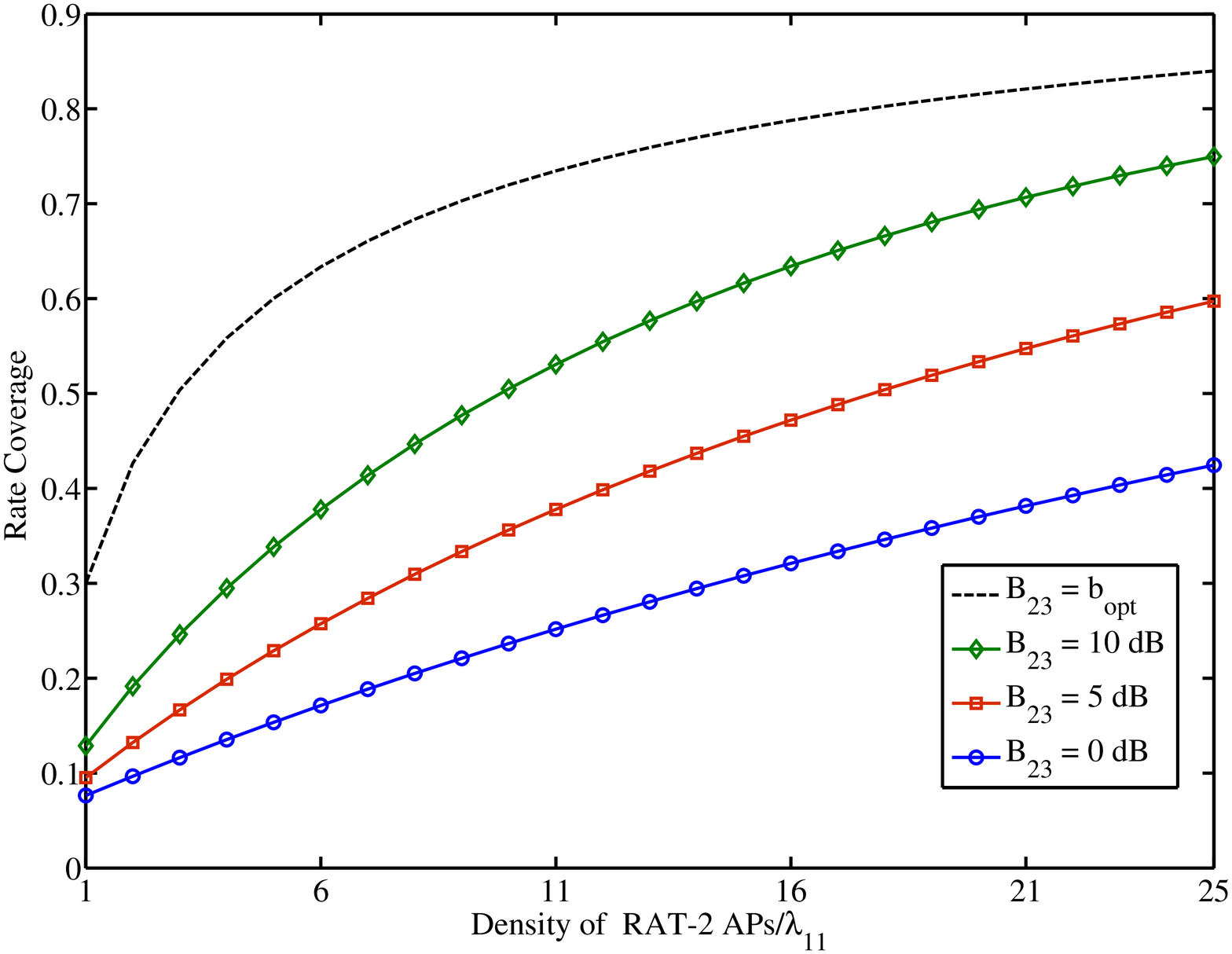}
		\caption{Effect of density of RAT-2 APs on rate coverage.}
	\label{fig:rcov_density}
\end{figure}
\begin{figure}
	\centering
		\includegraphics[width=\columnwidth]{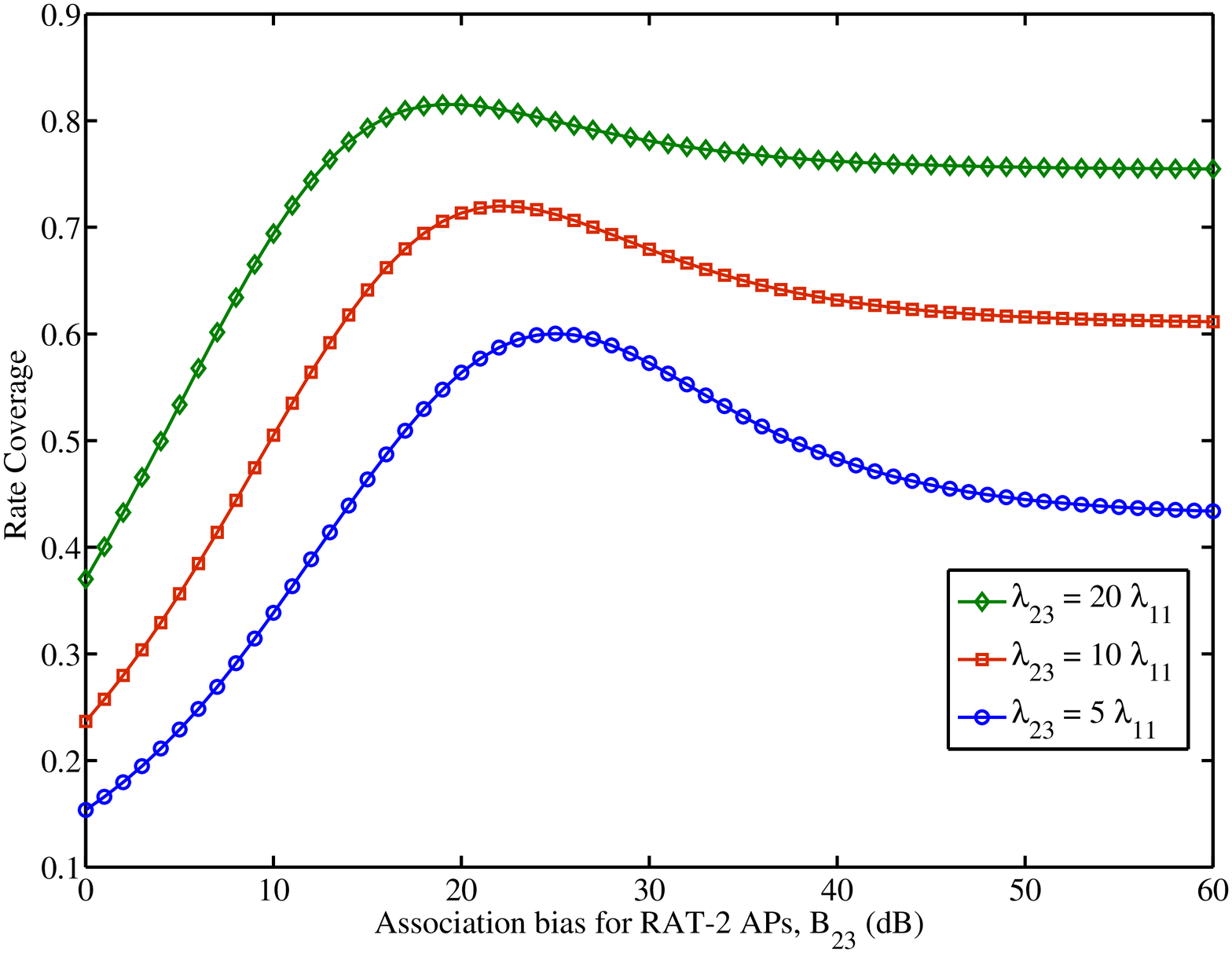}
		\caption{Effect of association bias for RAT-2 APs on rate coverage.}
	\label{fig:rcov_bias}
\end{figure}
\begin{figure}
	\centering
		\includegraphics[width=\columnwidth]{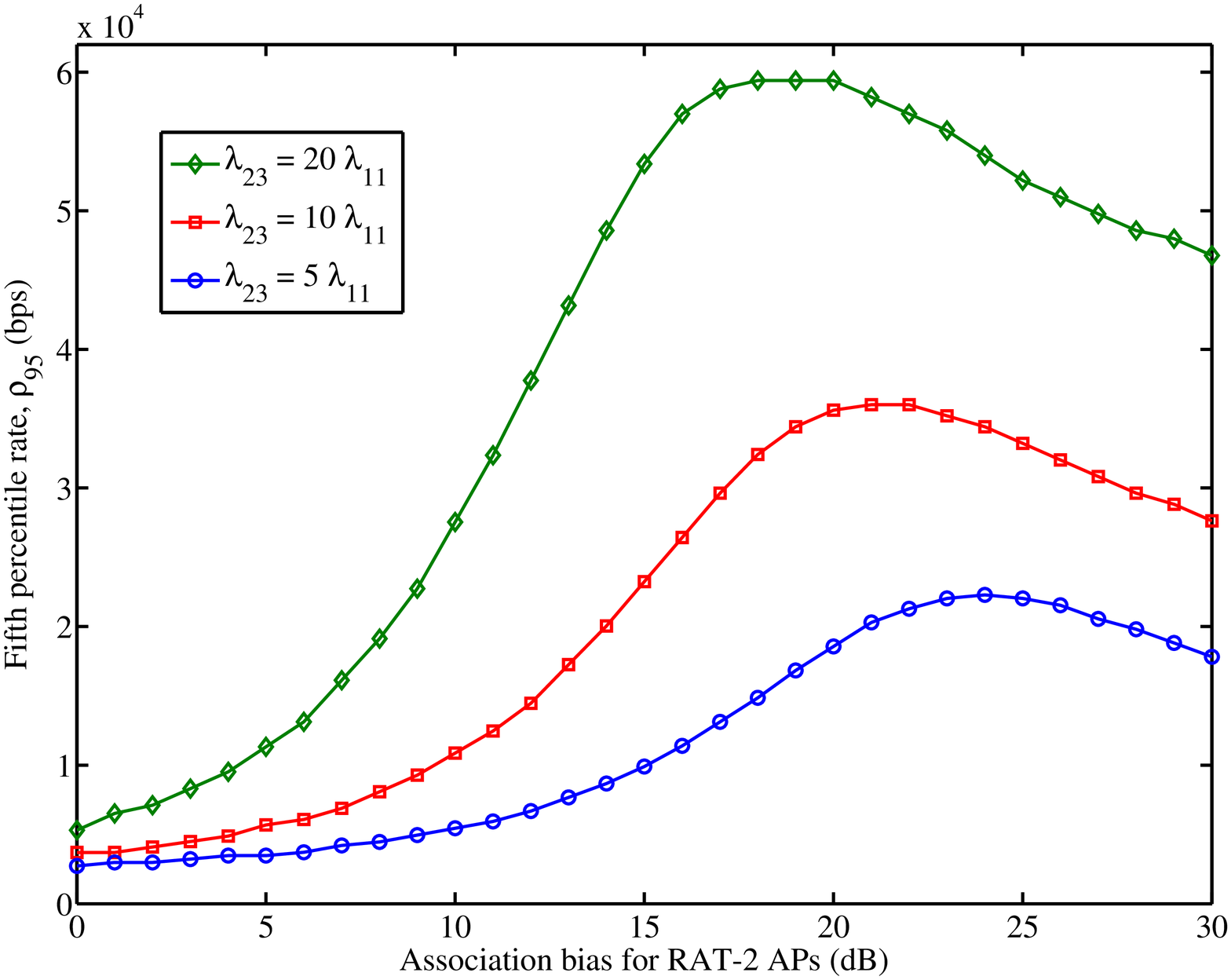}
		\caption{Effect of association bias for RAT-2 APs on $5^{\mathrm{th}}$ percentile rate with $\set{V}=\{(1,1);(2,3)\}$.}
	\label{fig:rcov_fiveile}
\end{figure}
\begin{figure}
	\centering
		\includegraphics[width=\columnwidth]{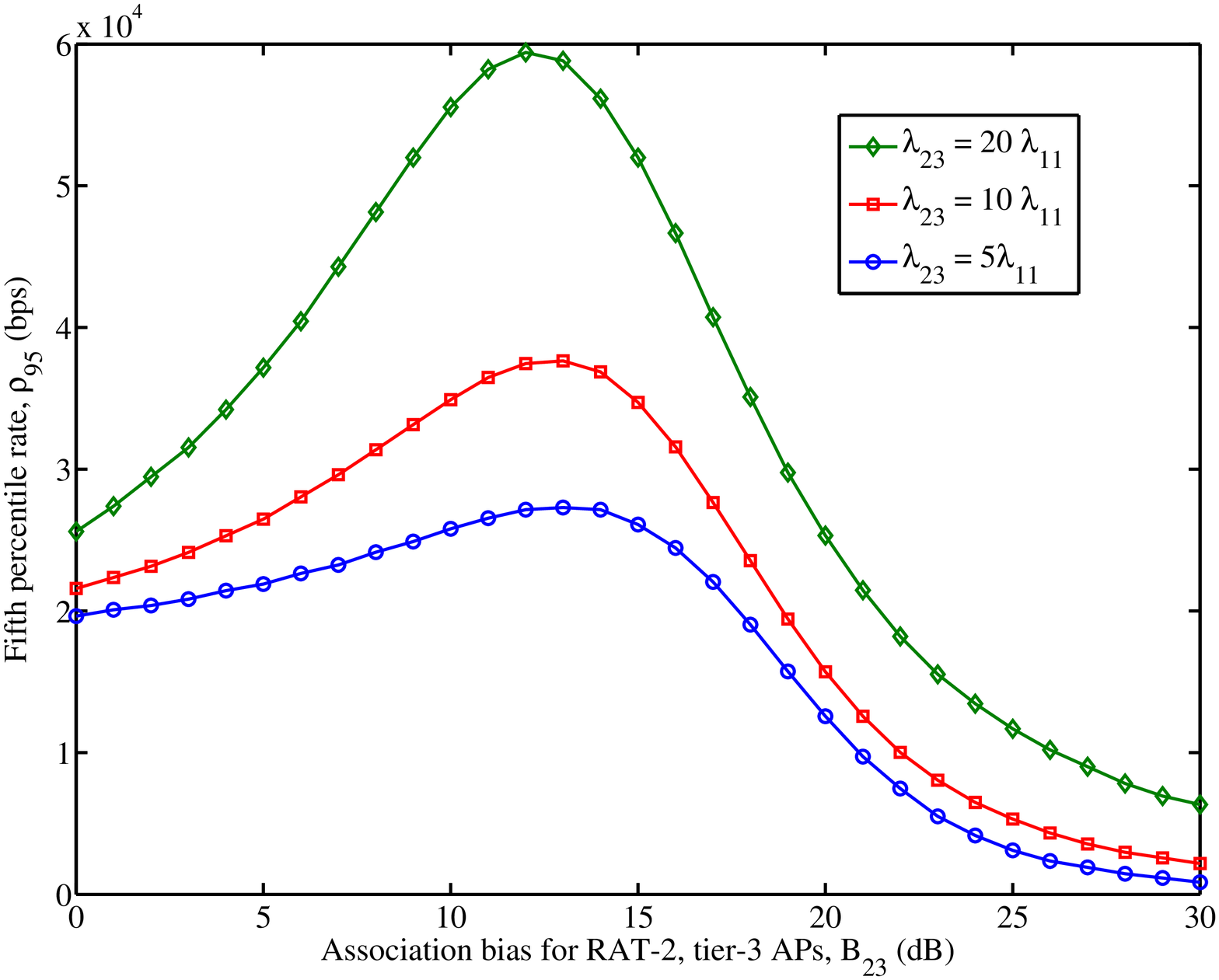}
		\caption{Effect of association bias for third tier of RAT-2 APs on $5^{\mathrm{th}}$ percentile rate with $\dnstysrat{1}{2}=\dnstysrat{2}{2}=5\dnstysrat{1}{1}$, $\bias{1}{2}=\bias{2}{2}=5$ dB.}
	\label{fig:rcov_fiveile_tworattwotier}
\end{figure}
\section{Conclusion}\label{sec:con}
In this paper, we presented a tractable model to analyze the effects of offloading in a  $M$-RAT   $K$-tier wireless heterogeneous network setting under a flexible association model. To the best of our knowledge, the presented work is the first to study rate coverage in the context of inter-RAT offload.
Using biased received power based association, it is shown that there exists an optimum percentage of the traffic that should be offloaded for maximizing the rate coverage  which in turn is dependent on user's QoS requirements and the resource condition at each available RAT  besides from the received signal power and load.  Investigating the coupling of AP queues induced by offloading, which has been ignored in this work, could be an interesting future extension.  Although the emphasis of this work has been on inter-RAT  offload, the framework can also be used to provide insights for  inter-tier offload within a RAT. Also, the area approximation for the association regions can be improved further by employing a non-linear approximation.
\begin{figure}
	\centering
		\includegraphics[width=\columnwidth]{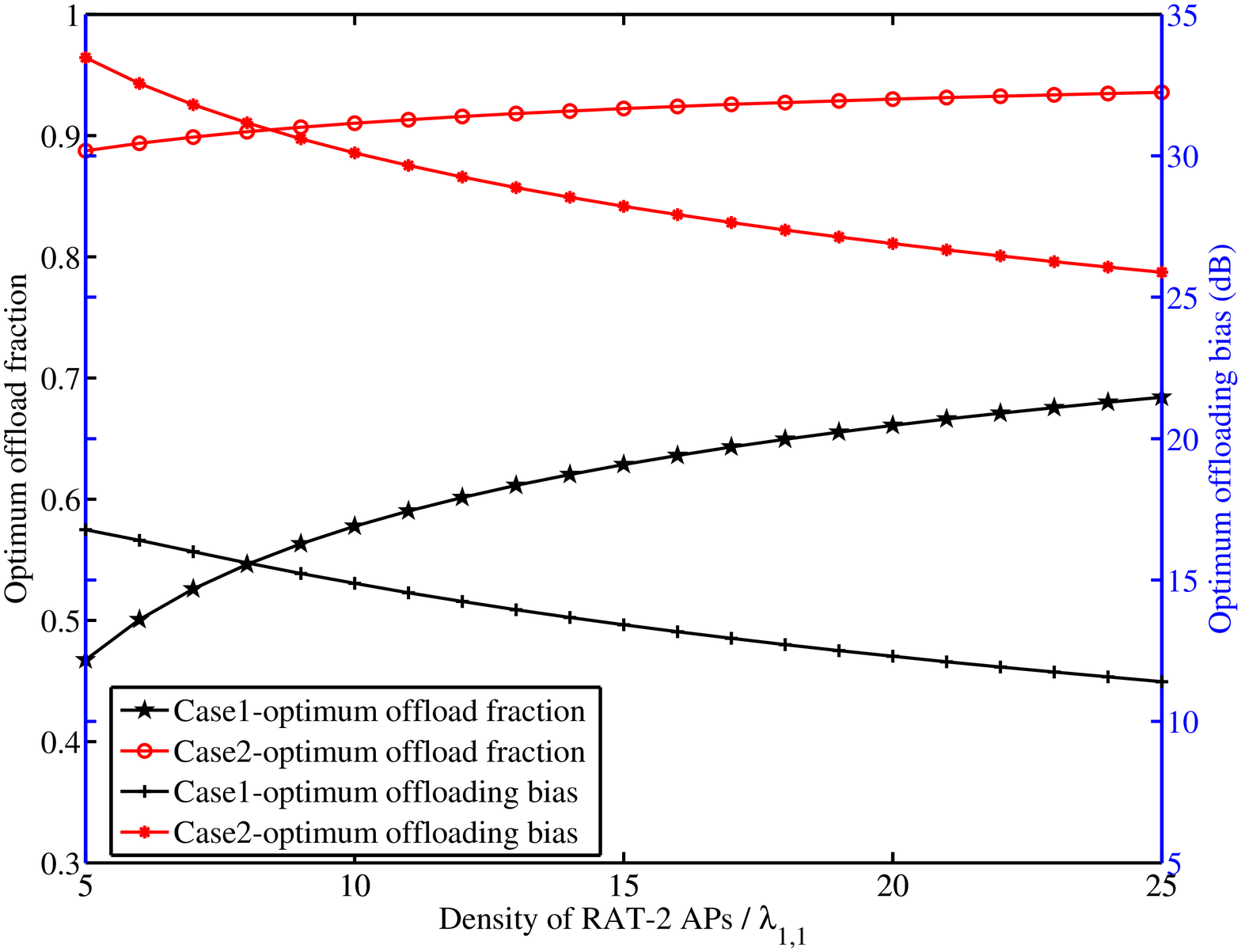}
		\caption{Effect of user's rate requirements and effective resources on the optimum association bias and optimum traffic offload fraction.}
	\label{fig:bias_trend}
\end{figure}
\appendices
%%%%%%%PROOF ASSOCIATION PROBABILITY %%%%%%%%%%%%%%
\section{}\label{sec:proofassocpr}
\begin{proof}[Proof of Lemma \ref{lem:aspr}\nopunct]
If $\passoc_{ij}$ is the association probability of a typical user with RAT-tier pair  $(i,j)$, then
\begin{equation}
\passoc_{ij} = \pr\left(\bigcap_{\substack{\mk \in \setopen{V}{}\\(m,k) \neq (i,j)}} \left\{\metric{i}{j}\NDIST{i}{j}^{-\ple{i}{j}}>\metric{m}{k}\NDIST{m}{k}^{-\ple{m}{k}}\right\}
\right),
\end{equation}
since $\NDIST{m}{k}$ denotes the distance to nearest  AP in $\PPPsrat{m}{k}$.
Thus
\small
\begin{align}
&\passoc_{ij} \overset{(a)}{=}  \prod_{\substack{(m,k) \in \setopen{V}{}\\(m,k) \neq (i,j)}} \pr\left( \metric{i}{j}\NDIST{i}{j}^{-\ple{i}{j}}>\metric{m}{k}\NDIST{m}{k}^{-\ple{m}{k}}
\right)\\
&=  \int\limits_{\ndistns>0} \prod_{\substack{(m,k) \in \setopen{V}{}\\(m,k) \neq (i,j)}} \pr\left( \NDIST{m}{k}>(\nmetric{m}{k})^{1/\ple{m}{k}}\ndistns^{1/\nple{m}{k}}
\right) f_{\NDIST{i}{j}}(\ndistns)\mathrm{d} \ndistns \label{eq:passoc1},
\end{align}
\normalsize
 where (a) follows from the independence of $\PPPsrat{m}{k}$, $\forall \mk \in \set{V}.$
Now
\begin{equation}\label{eq:nullpr}
\pr\left(\NDIST{m}{k} >\ndistns \right) =\pr\left(\PPPsrat{m}{k} \cap b(0,\ndistns) = \emptyset \right) = e^{-\pi \dnstysrat{m}{k}\ndistns^2},
\end{equation}
\normalsize
 where $b(0,z)$ is the Euclidean ball of radius $\ndistns$ centered at origin.
The probability distribution function $f_{\NDIST{m}{k}}(\ndistns)$ can be written as
\begin{align}
f_{\NDIST{m}{k}}(\ndistns)& = \frac{\mathrm{d}}{\mathrm{d}\ndistns}\{1- \pr(\NDIST{m}{k}>\ndistns)\}\nonumber\\
& = 2\pi\dnstysrat{m}{k}\ndistns\exp(-\pi\dnstysrat{m}{k}\ndistns^2),\,\, \forall \ndistns \geq 0.\label{eq:pdfndist}
\end{align}
Using (\ref{eq:passoc1}), (\ref{eq:nullpr}) and (\ref{eq:pdfndist})
\begin{multline}
\passoc_{ij} = 2\pi\dnstysrat{i}{j}\\
\times \int\limits_{\ndistns>0} \ndistns\exp\left( -\pi\sum_{\substack{(m,k)\in \setopen{V}{}\\\mk \neq (i,j)}}\dnstysrat{m}{k}(\nmetric{m}{k})^{2/\ple{m}{k}}\ndistns^{2/\nple{m}{k}}\right)\\\times\exp(-\pi\dnstysrat{i}{j}\ndistns^2)\mathrm{d} \ndistns,
\end{multline}
which gives (\ref{eq:aspr}).
\end{proof}

\section{}\label{sec:proofoloadpgf}
\begin{IEEEproof}[Proof of Lemma \ref{lem:oloadpgf}]
As a random user is more likely to lie in a larger association region then in a smaller association region, the distribution of the association area  of the tagged AP, $\area_{ij}^{'}$, is proportional to its area and can be written as
\begin{equation}
f_{\area_{ij}^{'}}(\areans) \propto \areans f_{\area_{ij}}(\areans).
\end{equation}
Using the normalization property of the distribution function and (\ref{eq:areadist}), the biased area distribution is
\small
\begin{align}
f_{\area_{ij}^{'}}(\areans)& = \frac{\areans f_{\area_{ij}}(\areans)}{\expect{\area_{ij}}}\label{eq:biaseareadist}
= \frac{3.5^{3.5}}{\Gamma(3.5)}\frac{\dnstysrat{i}{j}}{\passoc_{ij}}\left(\frac{\dnstysrat{i}{j}}{\passoc_{ij}}\areans\right)^{3.5}\exp\left(-3.5\frac{\dnstysrat{i}{j}}{\passoc_{ij}}\areans\right).
\end{align}
\normalsize
The location of the other users (apart from the typical user) in the association region of the tagged AP follows the reduced Palm distribution of $\PPPu$ which is the same as the original distribution since  $\PPPu$ is a PPP \cite[Sec. 4.4]{mecke_book}.
Thus, using Lemma \ref{lem:userpgf} and (\ref{eq:biaseareadist}), the PGF of the other users in the tagged AP is
\begin{align}
\pgf_{\oload{i}{j}}(z) & = \expect{\exp\left(\userdnsty\area_{ij}^{'}(z-1)\right)}\nonumber\\
& =\int\limits_{\areans>0}\exp\left(\userdnsty \areans(z-1)\right)\frac{3.5^{3.5}}{\Gamma(3.5)}\frac{\dnstysrat{i}{j}}{\passoc_{ij}}\left(\frac{\dnstysrat{i}{j}}{\passoc_{ij}}\areans\right)^{3.5}\nonumber\\
&\phantom{"="}\exp\left(-3.5\frac{\dnstysrat{i}{j}}{\passoc_{ij}}\areans\right)
\mathrm{d} \areans\\
&= 3.5^{4.5}\left(3.5+\frac{\userdnsty\passoc_{ij}}{\dnstysrat{i}{j}}(1-z)\right)^{-4.5}.
\end{align}
Using the PGF, the probability mass function can be derived as
\begin{align}
\pr\left(\load{i}{j}=n+1\right)&=\pr\left(\oload{i}{j}=n\right)=\frac{\pgf_{\oload{i}{j}}^{(n)}(0)}{n!}\nonumber\\
= \frac{3.5^{3.5}}{n!}\frac{\Gamma(n+4.5)}{\Gamma(3.5)}&\left(\frac{\userdnsty\passoc_{ij}}{\dnstysrat{i}{j}}\right)^n\times\left(3.5 + \frac{\userdnsty\passoc_{ij}}{\dnstysrat{i}{j}}\right)^{-(n+4.5)}.
\end{align}
For the second half of the proof, we use the property that the moments of a Poisson RV, $X \sim \text{Pois}(\lambda)$ (say), can be written in terms of Stirling numbers of the second kind, $\stirling{n}{k}$, as $\expect{X^n} = \sum_{k=0}^n \lambda^k\stirling{n}{k}$. Now
\begin{align}
\expect{\oload{i}{j}^n}&=\expect{\expect{\oload{i}{j}^n|\area_{ij}^{'}}}\\
 = & \expect{\sum_{k=0}^n (\userdnsty\area_{ij}^{'})^k\stirling{n}{k}}= \sum_{k=1}^n\userdnsty^k \stirling{n}{k}\expect{\area_{ij}^{'k}}.
\end{align}
Using (\ref{eq:biaseareadist}) and the area approximation (\ref{eq:area_approx})
\begin{align}
\expect{\area_{ij}^{'k}}&=\frac{\expect{\area_{ij}^{k+1}}}{\expect{\area_{ij}}}
= \frac{(\dnstysrat{i}{j}/\passoc_{ij})^{-(k+1)}\expect{\area^{k+1}(1)}}{(\dnstysrat{i}{j}/\passoc_{ij})^{-1}\expect{\area(1)}},
\end{align}
and thus
\[\expect{\oload{i}{j}^n} = \sum_{k=1}^n\left(\frac{\userdnsty\passoc_{ij}}{\dnstysrat{i}{j}}\right)^k \stirling{n}{k}\expect{\area^{k+1}(1)}.\]
\end{IEEEproof}

%%%%Proof of distance to nearest BS%%%%%%%%%%

\section{}\label{sec:proofndist}
\begin{IEEEproof}[Proof of Lemma \ref{lem:ndist}]
If $\NDISTc{i}{j}$ denotes the distance between the typical user and the tagged AP in $(i,j)$, then the distribution of $\NDISTc{i}{j}$ is the distribution of $\NDIST{i}{j}$ conditioned on the user being associated with $(i,j)$.
Therefore
\begin{align}
\pr(\NDISTc{i}{j}>\ndistnsc) &
=\pr\left(\NDIST{i}{j}>\ndistnsc| \text{ user is associated with } (i,j)\right)\\
&=\frac{\pr\left(\NDIST{i}{j}>\ndistnsc, \text{user is associated with  } (i,j)\right)}{\pr\left(\text{user is associated with } (i,j)\right)}.\label{eq:derivndist1}
\end{align}
Now using Lemma \ref{lem:aspr}
\begin{multline}\label{eq:derivndist3}
\pr\left(\NDIST{i}{j}>\ndistnsc,\text{ user is associated with }  (i,j)\right)
\\=2\pi\dnstysrat{i}{j}\int\limits_{\ndistns>\ndistnsc}\ndistns\exp\left(-\pi\sum_{\mk\in \setopen{V}{}}G_{ij}(m,k)\ndistns^{2/\nple{m}{k}}\right)\mathrm{d} \ndistns.
\end{multline}
Using (\ref{eq:derivndist1}) and (\ref{eq:derivndist3}) we get
\begin{multline}
\pr(\NDISTc{i}{j}>\ndistnsc)\\ = \frac{2\pi\dnstysrat{i}{j}}{\passoc_{ij}}\int\limits_{\ndistns>\ndistnsc}\ndistns\exp\left(-\pi\sum_{\mk\in \setopen{V}{}}G_{ij}(m,k)\ndistns^{2/\nple{m}{k}}\right)\mathrm{d} \ndistns,
\end{multline}
which leads to the PDF of $\NDISTc{i}{j}$
\begin{equation}
f_{\NDISTc{i}{j}}(\ndistnsc)= \frac{2\pi\dnstysrat{i}{j}}{\passoc_{ij}}\ndistnsc\exp\left(-\pi\sum_{\mk  \in \setopen{V}{}} G_{ij}(m,k)\ndistnsc^{2/\nple{m}{k}}\right).
\end{equation}
\end{IEEEproof}

%%%% PROOF: THEOREM%%%%%%%%
\section{}\label{sec:proofpcov}
\begin{IEEEproof}[Proof of Lemma \ref{lem:pcov}]
The $\SINR$ coverage of a user associated with an AP of $(i,j)$ is
\begin{equation}\label{eq:pcovproof0}
\pcov_{ij}(\SINRthresh_{ij})   = \int_{\ndistnsc>0} \pr(\SINR_{ij}(\ndistnsc) > \SINRthresh_{ij}) f_{\NDISTc{i}{j}}(\ndistnsc) \mathrm{d} \ndistnsc.
\end{equation}
Now $\pr(\SINR_{ij}(\ndistnsc)>\SINRthresh_{ij}) $ can be written as
\begin{align}
&\pr\left(\frac{\power{i}{j}h_{\ndistnsc} \ndistnsc^{-\ple{i}{j}}}{\sum_{k \in \set{V}_i} I_{ik} + \noisepower{i}}>\SINRthresh_{ij}\right)\\
&= \pr\left(h_\ndistnsc > \ndistnsc^{\ple{i}{j}}{\power{i}{j}}^{-1}\SINRthresh_{ij}\left\{\sum_{k \in \set{V}_i}I_{ik} +\noisepower{i} \right\}\right)\\
& =\expect{\exp\left(-\ndistnsc^{\ple{i}{j}}\SINRthresh_{ij}\power{i}{j}^{-1}\left\{\sum_{k \in \set{V}_i}I_{ik} + \noisepower{i}\right\}\right)}\\
&\overset{(a)}{=}\exp\left(-\frac{\SINRthresh_{ij}}{\SNR_{ij}(y)}\right)\prod_{k \in \set{V}_i} \cexpect{\exp\left(-\ndistnsc^{\ple{i}{j}}\SINRthresh_{ij}\power{i}{j}^{-1}I_{ik}\right)}{I_{ik}}\\
&=\exp\left(-\frac{\SINRthresh_{ij}}{\SNR_{ij}(y)}\right)\prod_{k \in \set{V}_i}\mgf_{I_{ik}}\left(\ndistnsc^{\ple{i}{j}}\SINRthresh_{ij}\power{i}{j}^{-1}\right)\label{eq:pcovproof1},
\end{align}
where $\SNR_{ij}(y) = \frac{\power{i}{j}\ndistnsc^{-\ple{i}{j}}}{\noisepower{i}}$ and (a) follows from the independence of $I_{ik}$ and $\mgf_{I_{ik}}(s)$ is the
the  moment-generating function (MGF) of the interference.
Expanding the interference term, the MGF of interference is given by
\begin{multline}
\mgf_{I_{ik}}(s)\\= \mathbb{E}_{\tierPPP{i}{k},\PPPcrat{i}{k},h_x,h_{x'}}\Bigg[\exp\Bigg(-s\power{i}{k}\Bigg\{\sum_{x\in \tierPPP{i}{k} \setminus o}  h_{x} x^{-\ple{i}{k}} \\
+ \sum_{x'\in \PPPcrat{i}{k}}  h_{x'} {x'}^{-\ple{i}{k}}\Bigg\}\Bigg)\Bigg]\end{multline}
\begin{align}
&\overset{(a)}{=} \cexpect{\prod_{x\in\tierPPP{i}{k}\setminus o}\mgf_{h_x}\left(s \power{i}{k}x^{-\ple{i}{k}}\right)}{\tierPPP{i}{k}}\nonumber\\
&\phantom{"="}\times \cexpect{\prod_{x'\in\PPPcrat{i}{k}}\mgf_{h_{x'}}\left(s\power{i}{k}x'^{-\ple{i}{k}}\right)}{\PPPcrat{i}{k}}\\
&\overset{(b)}{=}\exp\left(-2\pi\tierdnsty{i}{k}\int_{\ndist{i}{k}}^{\infty}\left\{1-\mgf_{h_x}\left(s \power{i}{k}x^{-\ple{i}{k}}\right)\right\}x\mathrm{d}x\right)\nonumber\\
&\times
\exp\left(-2\pi\dnstycrat{i}{k}\int_{0}^{\infty}\left\{1-\mgf_{h_{x'}}\left(s \power{i}{k}x'^{-\ple{i}{k}}\right)\right\}x'\mathrm{d}x'\right)\end{align}
%\\
\begin{align}
&\overset{(c)}{=}\exp\Bigg(-2\pi\tierdnsty{i}{k}\int_{\ndist{i}{k}}^{\infty}\frac{x}{1+(s\power{i}{k})^{-1} x^{\ple{i}{k}}}\mathrm{d}x\nonumber\\
&\phantom{"="}-2\pi\dnstycrat{i}{k}\int_{0}^{\infty}\frac{x'}{1+(s\power{i}{k})^{-1} x'^{\ple{i}{k}}}\mathrm{d}x'\Bigg),
\end{align}
where (a) follows from the independence of $\tierPPP{i}{k}$,$\PPPcrat{i}{k}$, $h_x$ and $h_x'$, (b) is obtained using the PGFL  \cite{mecke_book} of $\tierPPP{i}{k}$ and $\PPPcrat{i}{k}$, and (c) follows by using the MGF of an exponential RV with unit  mean.
In the above expressions,  $\ndist{i}{k}$ is the lower bound on distance of the closest open access interferer in $(i,k)$  which can be obtained by using (\ref{eq:association})
\begin{align}\label{eq:lowerboundndist}
\metric{i}{j}\ndistnsc^{-\ple{i}{j}}& = \metric{i}{k}\ndist{i}{k}^{-\ple{i}{k}} \text{ or }
\ndist{i}{k}= (\nmetric{i}{k})^{1/\ple{i}{k}}\ndistnsc^{1/\nple{i}{k}}.
\end{align}
Using change of variables  with $t= (s\power{i}{k})^{-2/\ple{i}{k}}x^2$, the integrals can be simplified as
\begin{align}
\int_{z_{ik}}^{\infty}&\frac{2x}{1+(s\power{i}{k})^{-1} x^{\ple{i}{k}}}\mathrm{d}x\nonumber\\
&= (s\power{i}{k})^{2/\ple{i}{k}}\int_{(s\power{i}{k})^{-2/\ple{i}{k}}z_{ik}^2}^{\infty}\frac{\mathrm{d}t}{1+t^{\ple{i}{k}/2}}\nonumber\\
& = \Z\left(s\power{i}{k},\ple{i}{k},z_{ik}^{\ple{i}{k}}\right),
\end{align}
and
\begin{align}
\int_{0}^{\infty}&\frac{2x}{1+(s\power{i}{k})^{-1} x^{\ple{i}{k}}}\mathrm{d}x
 = \Z\left(s\power{i}{k},\ple{i}{k},0\right),
\end{align}
where \[\Z(a,b,c)= a^{2/b}\int_{(\frac{c}{a})^{2/b}}^\infty \frac{\mathrm{d} u}{ 1+ u^{b/2}}.\]
This gives  the  MGF of interference
\begin{multline}\label{eq:mgfinterference}
\mgf_{I_{ik}}\left(s\right)
= \exp\Bigg(-\pi (s\power{i}{k})^{2/\ple{i}{k}}\\
\times\left\{\dnstysrat{i}{k}\Z\left(1,\ple{i}{k},\frac{z_{ik}^{\ple{i}{k}}}{s\power{i}{k}}\right)+\dnstycrat{i}{k}\Z\left(1,\ple{i}{k},0\right)\right\}\Bigg).
\end{multline}
Using $s=\ndistnsc^{\ple{i}{j}}\SINRthresh_{ij}\power{i}{j}^{-1}$  with $z_{ik}$ from (\ref{eq:lowerboundndist}) for MGF of interference in  (\ref{eq:pcovproof1}) we get
\begin{multline}
\pr(\SINR_{ij}(\ndistnsc)>\SINRthresh_{ij}) \\ = \exp\left(-\frac{\SINRthresh_{ij}}{\SNR_{ij}(y)}-\pi\sum_{k\in V_i}\ndistnsc^{2/\nple{i}{k}}D_{ij}\left(k,\SINRthresh_{ij}\right)\right),
\end{multline}
where
\footnotesize
 \begin{multline}
D_{ij}(k,\SINRthresh_{ij})  = \npower{i}{k}^{2/\ple{i}{k}}\left\{\tierdnsty{i}{k}\Z\left(\SINRthresh_{ij},\ple{i}{k},\npower{i}{k}^{-1}\nmetric{i}{k}\right)+\dnstycrat{i}{k}\Z\left(\SINRthresh_{ij},\ple{i}{k},0\right)\right\}.\nonumber\end{multline}
\normalsize
Using (\ref{eq:pcovproof0}) along with Lemma \ref{lem:ndist} gives
\footnotesize
\begin{multline}
\pcov_{ij}(\SINRthresh_{ij})  =\frac{2\pi\dnstysrat{i}{j}}{\passoc_{ij}}\int_{\ndistnsc>0}\ndistnsc\exp\Bigg(-\frac{\SINRthresh_{ij}}{\SNR_{ij}(y)}\\
-\pi\bigg\{\sum_{k\in \set{V}_i}D_{ij}(k,\SINRthresh_{ij})\ndistnsc^{2/\nple{i}{k}} + \sum_{\mk \in \setopen{V}{}} G_{ij}(m,k)\ndistnsc^{2/\nple{i}{k}}\bigg\}\Bigg)\mathrm{d}\ndistnsc.
\end{multline}
\normalsize
Using the law of total probability we get
\begin{equation}
\pcov = \sum_{(i,j)\in \setopen{V}{}}\pcov_{ij}(\SINRthresh_{ij})  \passoc_{ij},
\end{equation}
which gives the overall $\SINR$ coverage of a typical user.
\end{IEEEproof}

\section{}\label{sec:proofoptbias}
\begin{IEEEproof}[Proof of Proposition \ref{prop:optbias}]
In the described  setting $\SIR$ coverage can be written as
\begin{equation}\label{eq:pcovsimple}
\pcov= \sum_{(i,j)\in \setopen{V}{}} \frac{\dnstysrat{i}{j}}{ \sum_{k\in \set{V}_i}D_{ij}(k,\SINRthresh_{ij})  + \sum_{\mk \in \setopen{V}{}}G_{ij}(m,k)},
\end{equation}
and with $\set{V}=\{(1,q),(2,r)\}$, $\dnstysrat{2}{r} = a\dnstysrat{1}{q}$, and $\bias{2}{r} = b\bias{1}{q}$
\begin{align*}
\pcov & = \frac{\dnstysrat{1}{q}}{\dnstysrat{1}{q}\Z(\SINRthresh_{1q},\alpha,1)+ \dnstysrat{1}{q} +\dnstysrat{2}{r}(\npower{2}{r}\nbias{2}{r})^{2/\alpha}} \\
&\phantom{"="}+ \frac{\dnstysrat{2}{r}}{\dnstysrat{2}{r}\Z(\SINRthresh_{2r},\alpha,1)+ \dnstysrat{2}{r} +\dnstysrat{1}{q}(\npower{1}{q}\nbias{1}{q})^{2/\alpha}}\\
& = \frac{1}{\Z(\SINRthresh_{1q},\alpha,1)+1 +a(\npower{2}{r}b)^{2/\alpha}} \\
& \phantom{"="}+ \frac{1}{\Z(\SINRthresh_{2r},\alpha,1)+ 1 +\frac{1}{a(\npower{2}{r}b)^{2/\alpha}}}.
\end{align*}
The gradient of $\pcov$ with respect to association bias $\nabla_b\pcov$ is zero  at
\begin{align*}
b_\mathrm{opt} & = \arg \max_{b} \Bigg\{\left(\Z(\SINRthresh_{1q},\alpha,1)+1 +a(\npower{2}{r}b)^{2/\alpha}\right)^{-1}\\
& + \left(\Z(\SINRthresh_{2r},\alpha,1)+ 1 +\frac{1}{a(\npower{2}{r}b)^{2/\alpha}}\right)^{-1}\Bigg\}\\
& = \frac{\power{1}{q}}{\power{2}{r}}\left(\frac{\Z(\SINRthresh_{1q},\alpha,1)}
{a\Z(\SINRthresh_{2r},\alpha,1)}\right)^{\alpha/2}.
\end{align*}
With algebraic manipulation, it can be shown that for all $b> b_\mathrm{opt}$  $\nabla_b\pcov<0$ and for all $b< b_\mathrm{opt}$ $\nabla_b\pcov>0$ and  hence $\pcov$ is strictly quasiconcave in $b$ and $b_\mathrm{opt} $ is the unique mode.
Using Lemma \ref{lem:aspr}, the optimal traffic offload fraction  is obtained as
\begin{align}
\passoc_2 &= \frac{\dnstysrat{2}{r}}{G_{2r}(r)} = {a}\left\{a+ \left(\frac{\power{1}{q}}{\power{2}{r}b_\mathrm{opt}}\right)^{2/\alpha}\right\}^{-1}\nonumber\\
&= \frac{\Z(\SINRthresh_{1q},\alpha,1)}{\Z(\SINRthresh_{2r},\alpha,1)+\Z(\SINRthresh_{1q},\alpha,1)}.
\end{align}
The corresponding $\SIR$ coverage can then obtained by substituting the optimal bias value in (\ref{eq:pcovsimple}).
\end{IEEEproof}

\bibliographystyle{ieeetr}
\bibliography{IEEEabrv,refoffload}

\end{document}